\documentclass[sigconf]{acmart}
\renewcommand\footnotetextcopyrightpermission[1]{} 
\usepackage{booktabs} 

\usepackage{stmaryrd}
\usepackage{enumitem}
\usepackage{tikz}

\usepackage[textwidth=3cm, shadow, obeyFinal]{todonotes}

\makeatletter
\newcommand{\oset}[3][0ex]{%
  \mathrel{\mathop{#3}\limits^{
    \vbox to#1{\kern-2\ex@
    \hbox{$\scriptstyle#2$}\vss}}}}
\makeatother

\DeclareMathOperator{\eventually}{\mathbf{F}}
\DeclareMathOperator{\globally}{\mathbf{G}}

\DeclareMathOperator{\nextx}{\mathbf{X}}

\newcommand*{\counting}{\mathop{\mathbf{C}}}
\newcommand*{\until}{\mathbin{\mathbf{U}}}

\newcommand*{\release}{\mathbin{\mathbf{R}}}

\newcommand*{\ltl}{\textup{\textmd{\textsf{LTL}}}}

\newcommand*{\cemitl}{\textup{\textmd{\textsf{CEMITL}}}}
\newcommand*{\emitl}{\textup{\textmd{\textsf{EMITL}}}}
\newcommand*{\midl}{\textup{\textmd{\textsf{MIDL}}}}

\newcommand*{\ctmitl}{\textup{\textmd{\textsf{CTMITL}}}}

\newcommand*{\mitl}{\textup{\textmd{\textsf{MITL}}}}
\newcommand*{\eecl}{\textup{\textmd{\textsf{EECL}}}}
\newcommand*{\ecl}{\textup{\textmd{\textsf{ECL}}}}

\newcommand*{\mtl}{\textup{\textmd{\textsf{MTL}}}}
\newcommand*{\foone}{\textup{\textmd{\textsf{FO[$<, +1$]}}}}
\newcommand*{\fo}{\textup{\textmd{\textsf{FO[$<$]}}}}
\newcommand*{\tre}{\textup{\textmd{\textsf{TRE}}}}

\newcommand*{\nfa}{\textup{\textmd{\textsf{NFA}}}}
\newcommand*{\dfa}{\textup{\textmd{\textsf{DFA}}}}
\newcommand*{\ta}{\textup{\textmd{\textsf{TA}}}}
\newcommand*{\dta}{\textup{\textmd{\textsf{DTA}}}}

\newcommand*{\ocata}{\textup{\textmd{\textsf{OCATA}}}}

\newcommand*\A{\mathcal A} 
\newcommand*\B{\mathcal B} 
\newcommand*\C{\mathcal C} 
\newcommand*\D{\mathcal D} 
\newcommand*\E{\mathcal E}
\newcommand*\F{\mathcal F}
\newcommand*\transitions{\Delta}
\newcommand*\atransitions{\delta}
\newcommand*\Guards{\mathcal G}
\newcommand*\R{\mathbb R}

\newcommand*\N{\mathbb N}
\newcommand*\sem[1]{\ensuremath{\llbracket#1\rrbracket}}

\newcommand*\Lang{\mathcal L}
\renewcommand*\phi{\varphi}
\newcommand*\AP{\textup{\textmd{\textsf{AP}}}}

\newcommand*\proj{\textup{\textmd{\textsf{proj}}}}

\newcommand*{\examplename}{Ex.}
\renewcommand*{\figurename}{Fig.}

\usetikzlibrary{positioning,automata,arrows,hobby,decorations.pathreplacing,fit}

\tikzstyle{every state}=[minimum size=1.5em]

\tikzset{>=latex,auto,node distance=2.5cm, every
  loop/.style={looseness=6}, initial text={}, inner sep=0.5mm,
  loopright/.style={loop,looseness=6,out=35, in=-35},
  loopleft/.style={loop,looseness=6,out=145, in=215},
  loopabove/.style={loop,looseness=6,out=125, in=55},
  loopbelow/.style={loop,looseness=6,out=-125, in=-55}, }

\setcopyright{rightsretained}

\acmConference[HSCC'19]{22nd ACM International Conference on 
Hybrid Systems: Computation and Control}{April 2019}{Montr\'{e}al, Canada} 
\acmYear{2019}
\copyrightyear{2019}

\begin{document}
\title{Revisiting Timed Logics with Automata Modalities}

\author{Hsi-Ming Ho}
\orcid{0000-0003-0387-4857}
\affiliation{%
  \institution{University of Cambridge}
  \country{United Kingdom}
}
\email{hsi-ming.ho@cl.cam.ac.uk}


\begin{abstract}
It is well known that (timed) $\omega$-regular properties such as `$p$ holds at every even position'
	and `$p$ occurs at least three times within the next 10 time units'
	cannot be expressed in Metric Interval Temporal Logic (\mitl{})
	and Event Clock Logic (\ecl{}).
	A standard remedy to this deficiency is to extend these
	with modalities defined in terms of automata. 
	In this paper, we show that the logics $\emitl_{0, \infty}$  
	(adding \emph{non-deterministic finite automata} modalities
	into the fragment of \mitl{} with only lower- and upper-bound constraints)
	and \eecl{} (adding automata modalities into \ecl{})
	are already as expressive as
	\emitl{} (full \mitl{} with automata modalities).
	In particular, the satisfiability and model-checking problems for
	$\emitl_{0, \infty}$ and \eecl{}
	are $\mathrm{PSPACE}$-complete, whereas the same problems for $\emitl$
	are $\mathrm{EXPSPACE}$-complete.
	We also provide a simple translation from $\emitl_{0, \infty}$
	to \emph{diagonal-free} timed automata, which
	enables practical satisfiability and model checking based on off-the-shelf tools. 
\end{abstract}

%
%



\maketitle

\section{Introduction} 


\paragraph{Timed logics} In the context of real-time systems verification, it is natural and desirable to add \emph{timing constraints} to
\emph{Linear Temporal Logic} (\ltl{})~\cite{Pnueli1977}
to enable reasoning about timing behaviours of such systems.
For instance, one may write $\phi_1 \until_I \phi_2$ to assert
that $\phi_1$ holds until a `witness' point where $\phi_2$ holds,
and the time difference between now and that point lies within the
\emph{constraining interval} $I$.
The resulting logic, \emph{Metric Temporal Logic} (\mtl{})~\cite{Koy90},
can be seen as a fragment of \emph{Monadic First-Order Logic of Order and Metric} (\foone{})~\cite{AluHen93}, the timed counterpart of the classical \emph{Monadic First-Order Logic of Order} (\fo{}).
There are, nonetheless, some loose ends in this analogy. For instance, while \ltl{} is as expressive as \fo{}~\cite{Kamp1968, Gabbay1980}, it is noted early on 
that certain `non-local' timing properties in \foone{}, albeit being very simple,
cannot be expressed in timed temporal logics like \mtl{}~\cite{AluHen94}. 
As a concrete example, the property `every $p$-event is followed by a $q$-event and, later, an $r$-event within the next 10 time units', written as the \foone{}  formula
\begin{equation}\label{eq:pqr}
\forall x \, (p(x) \Rightarrow \exists y \, \Big(q(y) \land \exists z \, \big(r(z) \land x \leq y \leq z \leq x + 10)\big)\Big)
\end{equation}
is not expressible in \mtl{}---indeed, no `finitary' extension of \mtl{} can be \emph{expressively complete}
for \foone{}~\cite{Rabinovich2007}.\footnote{(\ref{eq:pqr}) can, however, be expressed in \mtl{} if the  \emph{continuous} semantics of the logic is adopted
or past modalities are allowed; see~\cite{Bouyer2010} for details.}
A more serious practical concern 
is that the satisfiability problem for \mtl{} is undecidable~\cite{AluHen93, Ouaknine2006}. For this reason, research efforts have been 
focused on fragments of \mtl{} with decidable satisfiability, most notably
\emph{Metric Interval Temporal Logic} (\mitl{}), the fragment of \mtl{} in which
`punctual' constraining intervals are not allowed~\cite{AluFed96}.
In particular, \mitl{} formulae can be effectively translated into \emph{timed automata} (\ta{s})~\cite{AluDil94},
giving practical $\mathrm{EXPSPACE}$ decision procedures for its \emph{satisfiability} and \emph{model-checking} problems~\cite{BriEst14,
BriGee17, BriGee17b}. 

\paragraph{Automata modalities}
It is well known that properties that are necessarily \emph{second order} (e.g.,~`$p$ holds at all even positions') cannot be expressed in \ltl{} or \mitl{}.
Fortunately, it is possible to add \emph{automata modalities} into
\ltl{} 
at no additional computational cost~\cite{Wolper1994, Sistla1985}.
In timed settings, the logic obtained from \mitl{} by adding
\emph{time-constrained} automata modalities defined by non-deterministic finite automata (\nfa{s})
is called \emph{Extended Metric Interval Temporal Logic} (\emitl{})~\cite{Wilke1994}.
From a theoretical point of view, \emitl{} is a \emph{fully decidable}
formalism (i.e.~constructively closed under all Boolean operations and with decidable satisfiability~\cite{HenRas98})
whose class of timed languages strictly contains
that of \mitl{} and B\"uchi automata.\footnote{A very recent paper of Krishna, Madnani, and Pandya~\cite{KriMad18} showed
that this class admits some alternative characterisations (namely, a syntactic fragment of \ocata{s} and
a timed monadic second-order logic).}
In practice, it can be argued that automata modalities are natural, easy-to-use
extensions of the usual \mitl{} modalities. %
They also allow  properties like (\ref{eq:pqr}), which often
emerge in application domains like healthcare and automotive engineering,
to be written as specifications.
\begin{example}[\cite{Abbas17}]\label{ex:icd}
Discrimination algorithms are implemented in implantable cardioverter defibrillators (ICDs)
to detect potentially dangerous heartbeat patterns. As a simple example, one may want to check
whether \emph{the number of heartbeats in one minute is between $120$ and $150$}. This can be expressed
as the \ctmitl{}~\cite{KriMad16} formula $\counting_{[0, 59]}^{\geq 120} p \wedge \counting_{[0, 59]}^{\leq 150} p$
where $p$ denotes a peak in the cardiac signal. The \emph{counting modalities}
$\counting^{\sim k}_{I}$ (where $0 \in I$, which is the case here),
as well as~(\ref{eq:pqr}), be expressed straightforwardly in terms of automata.
\end{example}
\begin{example}[adapted from~\cite{OpenScenario16}]\label{ex:overtaking}
In autonomous driving, one may want to specify that 
\emph{a car overtaking another from the left must be done in $10$ seconds}.
Suppose the lane on the left is empty and the events are sampled sufficiently frequently (say $5$ms), this can be expressed as 
the \emitl{} formula $\A_{[0, 10]} (\texttt{TTC > 4}, \dots)$
(see \figurename~\ref{fig:vehicles} and \figurename~\ref{fig:overtaking})
where $\texttt{TTC}$ is the time to collision, 
$\texttt{dist}$ is the longitudinal distance between the two vehicles, and 
$\texttt{to\_left}$, $\texttt{to\_right}$ are the actions for merging to the left/right lane---these are taken immediately after $\texttt{TTC <= 4}$
and $\texttt{dist >= 5}$, respectively.

\end{example}



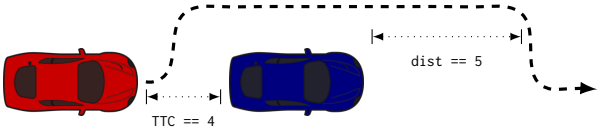
\begin{figure}[h]
\centering
	\begin{tikzpicture}
\node[inner sep=0pt] (redcar) at (0,0)
    {\includegraphics[width=.1\textwidth]{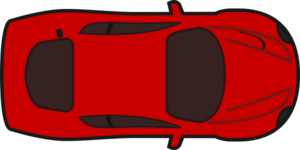}};
\node[inner sep=0pt] (bluecar) at (3,0)
    {\includegraphics[width=.1\textwidth]{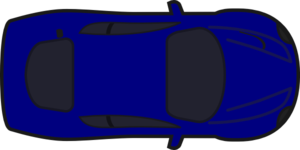}};
\node[inner sep=0pt] (bluecar) at (8,0) {};

\draw[dashed,->, very thick] (1,0) to[curve through={(1.1,0) (1.5,0.9) (2,1) (2.5,1) (3,1) (3.5,1) (4,1) (4.5, 1) (5,1) (5.5,1) (6,0.9) (6.4,-.1) (6.5, -.1)}] (7,-.1);

\draw[|<->|][dotted] (1,-.2) -- (2,-.2) node[midway,below=0.2cm] {\scriptsize $\texttt{TTC == 4}$};

\draw[|<->|][dotted] (4,.6) -- (6,.6) node[midway,below=0.2cm] {\scriptsize $\texttt{dist == 5}$};

	\end{tikzpicture}
\caption{The red car overtakes the blue car from the left.}
\label{fig:vehicles}
\end{figure}

\begin{figure}[h]
\centering
	\scalebox{.75}{\begin{tikzpicture}[node distance = 4cm]
		\node[initial left ,state] (0) {};
		\node[state, right=2.6cm of 0] (1) {};
		\node[state, right=1.3cm of 1] (2) {};
		\node[state, right=2.6cm of 2] (3) {};
		\node[state, accepting, right=1.3cm of 3] (4) {};
		\path
		(0) edge[loopabove, ->] node[above, align=center]{\scriptsize $\texttt{TTC > 4}$} (0)
		(2) edge[loopabove, ->] node[above, align=center]{\scriptsize $\texttt{dist < 5}$} (2)

		(0) edge[->] node[above] {\scriptsize $\texttt{TTC <= 4}$} (1)

		(2) edge[->] node[above] {\scriptsize $\texttt{dist >= 5}$} (3)
		(3) edge[->] node[above] {\scriptsize $\texttt{to\_right}$} (4)
		(1) edge[->] node[above] {\scriptsize $\texttt{to\_left}$} (2);
	\end{tikzpicture}}
\caption{$\mathcal{A}$ in Example~\ref{ex:overtaking}.}
\label{fig:overtaking}
\end{figure}
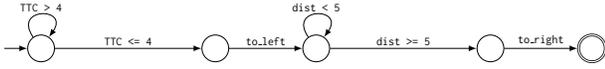

Compared with \ltl{} and \mitl{}, however, translating \emitl{} 
into \ta{s} is considerably more challenging.
The original translation by Wilke~\cite{Wilke1994}
is non-elementary and thus not suitable for practical purposes.
Krishna, Madnani, and Pandya~\cite{KriMad17}
showed that any \emitl{} formula can be encoded into an \mitl{} formula of doubly exponential size (which can then be translated into a \ta{}),
but this does not match the $\mathrm{EXPSPACE}$ lower bound inherited from \mitl{}.
More recently, Ferr\`ere~\cite{Ferrere18} proposed an asymptotically optimal construction from \midl{} (\emph{Metric Interval Dynamic Logic},
which is strictly more expressive and subsumes \emitl{}) formulae
to \ta{s}, but it is very complicated and relies heavily on the use of \emph{diagonal constraints} (i.e.~comparison between clocks) which are, in general, not preferred in practice~\cite{Bouyer03, Bouyer2005, Gastin2018} and not well-supported by existing model checkers.\footnote{It is possible to obtain a diagonal-free \ta{} from an \emitl{} formula
by first applying the construction in~\cite{Ferrere18} and then removing the diagonal constraints~\cite{Bouyer05}. 
This, however, is expensive and difficult to implement.}

\paragraph{Contributions}
We consider a simple fragment of \emitl{}, which we call $\emitl_{0, \infty}$, obtained
by allowing only lower- and upper-bound constraining intervals (e.g.,~$[0, a)$ and $(b, \infty)$)
and \eecl{}~\cite{Raskin1999} (adding automata modalities to \emph{Event Clock Logic} \ecl{}).
The satisfiability and model-checking problems for $\emitl_{0, \infty}$ and \eecl{} are much cheaper
than that of \emitl{} ($\mathrm{PSPACE}$-complete vs $\mathrm{EXPSPACE}$-complete).
Moreover, we show that they are already as expressive as full \emitl{}---this is in sharp contrast with the situation for
`vanilla' $\mitl_{0, \infty}$/$\ecl$ and \mitl{}, where the latter is strictly more expressive
when interpreted over timed words~\cite{HenRas98, Raskin1999}---making them
\emph{expressive} yet \emph{tractable} real-time specification formalisms.
We then show that $\emitl_{0, \infty}$
admits a much simpler translation into \ta{s}. 
Specifically, by effectively decoupling the timing
and operational aspects of automata modalities,
overlapping obligations imposed by a single automaton subformula
can be handled in a purely fifo manner
with a set of \emph{sub-components} (each of which is a simple one-clock \ta{}
with a polynomial-sized symbolic representation),
avoiding the use of diagonal constraints altogether.\footnote{For simplicity
we focus on logics with only future modalities, but our results readily carry over to the versions with both future and past modalities,
thanks to the compositional nature of our construction (cf. e.g.,~\cite{Kesten1998, Nickovic2008}).}
This makes our construction better suited to be implemented
to work with existing highly efficient algorithmic back ends (e.g.,~\textsc{Uppaal}~\cite{Behrmann2006} and
\textsc{LTSmin}~\cite{KanLaa15}).

%

\paragraph{Related work}
The idea of extending \ltl{} to capture the full class
of $\omega$-regular languages dates back to the seminal works of Clarke, Sistla, Vardi, and
Wolper~\cite{Wolper1983, Wolper1994, Sistla1985, Sistla1985b} in the early 1980s.
In particular, it is shown that \ltl{} with \nfa{} modalities---which
essentially underlies various industrial specification languages
like ForSpec~\cite{Armoni02} and PSL~\cite{Eisner06}---are expressively equivalent to B\"uchi automata,
yet the model-checking and satisfiability problems remain
$\mathrm{PSPACE}$-complete, same as \ltl{}.\footnote{There are other ways to extend \ltl{} to achieve
$\omega$-regularity, e.g., adding monadic second-order quantifiers
($\textup{\textmd{\textsf{QPTL}}}$~\cite{Sistla1985}) or
least/greatest fixpoints ($\textup{\textmd{\textsf{$\mu$LTL}}}$~\cite{Banieqbal1987, Vardi1987}).
These formalisms unfortunately suffer from higher complexity or less readable syntax.}
Our approach generalises the construction in~\cite{Wolper1994}
in the case of finite acceptance.

Henzinger, Raskin, and Schobbens~\cite{HenRas98, Raskin1999}
proved a number of analogous results in timed settings; in particular, 
they showed that in the \emph{continuous} semantics (i.e.~over finitely variable \emph{signals}),
(i) $\mitl_{0, \infty}$ and \ecl{} are as expressive as \mitl{}, and (ii)
the fragment of $\emitl$ with \emph{unconstrained} automata modalities 
is as expressive as \emph{recursive event-clock automata},
and the verification problems for this fragment can be solved in $\mathrm{EXPSPACE}$.
Our results can be seen as counterparts 
in the \emph{pointwise} semantics (i.e.~over \emph{timed words}). 

Besides satisfiability and model checking, extending timed logics with
automata or regular expressions is also a topic of great interest in
\emph{runtime verification}.
Basin, Krsti\'{c}, and Traytel~\cite{Basin2017}
showed that \mtl{} with time-constrained regular-expression modalities admits an efficient runtime monitoring procedure in a pointwise, integer-time setting.
A very recent work of Ni\v{c}kovi\'c, Lebeltel, Maler, Ferr\`ere, and Ulus~\cite{Nickovic2018}
considered a similar extension of \mitl{} with \emph{timed regular expressions}
(\tre{})~\cite{Asarin1997, Asarin2002} in the context of monitoring and analysis of Boolean and real-valued signals. 

\section{Timed logics and automata}

\paragraph{Timed languages}
A \emph{timed word} over a finite alphabet $\Sigma$ is an infinite sequence
of \emph{events} $(\sigma_i,\tau_i)_{i \geq 1}$ over
$\Sigma \times \R_{\geq 0}$ with $(\tau_i)_{i\geq 1}$ a non-decreasing
sequence of non-negative real numbers such that for each $r \in \R_{\geq 0}$,
there is some $j \geq 1$ with $\tau_j \geq r$ (i.e.~we require all timed words
to be `\emph{non-Zeno}').
We denote by $T\Sigma^\omega$ the set of all timed words over $\Sigma$. A \emph{timed language} is a
subset of $T\Sigma^\omega$.

\paragraph{Extended timed logics} 
A \emph{non-deterministic finite automaton} (\nfa{}) over $\Sigma$
is a tuple $\A = \langle \Sigma, S, s_0,\transitions, F \rangle$
where $S$ is a finite set of locations, $s_0 \in S$ is the initial location,
$\transitions \subseteq S \times \Sigma \times S$ is the transition relation,
and $F$ is the set of final locations.
We say that $\A$ is \emph{deterministic} (a \dfa{}) iff for each $s \in S$
and $\sigma \in \Sigma$, $| \{ (s, \sigma, s') \mid (s, \sigma, s') \in \transitions \} | \leq 1$. 
A \emph{run} of $\A$ on $\sigma_1 \dots \sigma_n \in \Sigma^+$
(without loss of generality, we only consider runs of automata modalities over \emph{nonempty} finite words in this paper) is a
sequence of locations $s_0 s_1 \dots s_n$ where 
there is a transition $(s_i,\sigma_{i+1},s_{i+1}) \in \transitions$
for each $i$, $0 \leq i < n$.  A run of $\A$ is \emph{accepting} iff
it ends in a final location. A finite word is \emph{accepted} by $\A$
iff $\A$ has an accepting run on it.
We denote by $\sem{\A}$ the set of finite words accepted by $\A$.

\emph{Extended Metric Interval Temporal Logic} (\emitl{}) formulae over
a finite set of atomic propositions $\AP$
are generated by  
\begin{displaymath}
  \phi := \top \mid p \mid \phi_1 \land \phi_2 \mid \neg\phi \mid \A_I(\phi_1, \dots, \phi_n)
\end{displaymath}
where $p\in\AP$, $\A$ is an \nfa{} over the $n$-ary alphabet $\{ 1, \dots, n \}$,
and $I \subseteq \R_{\geq 0}$ is a non-singular interval with endpoints in $\N_{\geq 0} \cup\{\infty\}$.\footnote{For notational simplicity,
we will occasionally use $\phi_1$,~\dots, $\phi_n$ directly as transition labels (instead of $1$, \dots, $n$).}
As usual, we omit the subscript $I$ when $I = [0, \infty)$ and write pseudo-arithmetic expressions for lower
or upper bounds, e.g.,~`$< 3$' for $[0, 3)$.
We also omit the arguments $\phi_1$,~\dots, $\phi_n$
and simply write $\A_I$, if clear from the context.
Following~\cite{AluHen93, AluHen94, Wilke1994, OuaWor07}, we consider
the pointwise semantics of \emitl{} and interpret
formulae over timed words: given an \emitl{} formula $\phi$ over $\AP$, 
a timed word $\rho=(\sigma_1,\tau_1)(\sigma_2,\tau_2)\dots$ over $\Sigma_\AP = 2^\AP$ and
a \emph{position} $i \geq 1$,
\begin{itemize}
\item $(\rho, i)\models \top$;
\item $(\rho, i)\models p$ iff $p\in \sigma_i$;
\item $(\rho, i)\models \phi_1\land \phi_2$ iff $(\rho,i)\models\phi_1$ and 
$(\rho,i)\models\phi_2$;
\item $(\rho,i)\models\neg\phi$ iff $(\rho,i)\not\models\phi$;
\item $(\rho,i)\models \A_I(\phi_1, \dots, \phi_n)$ iff there exists
  $j\geq i$ such that (i) $\tau_j-\tau_i\in I$ and (ii) there is an
accepting run of $\A$ on $a_i \dots a_j$ where $a_\ell \in \{1, \dots, n\}$
and $(\rho, \ell) \models \phi_{a_\ell}$  for each $\ell$, $i \leq \ell \leq j$.\footnote{Note that
it is possible for $(\rho,i)\models \A_I(\phi_1, \dots, \phi_n)$ and $(\rho,i) \models \A^c_I(\phi_1, \dots, \phi_n)$,
where $\A^c$ is the complement of $\A$, to hold simultaneously.}
\end{itemize}
The other Boolean operators are defined as usual:
$\bot \equiv \neg\top$,
$\phi_1\lor\phi_2 \equiv \neg(\neg\phi_1\land\neg\phi_2)$,
and $\phi_1\Rightarrow \phi_2 \equiv \lnot\phi_1\lor \phi_2$.
We also define the dual automata modalities
$\tilde{\A}_I(\phi_1, \dots, \phi_n) \equiv \neg \A_I(\neg \phi_1, \dots, \neg \phi_n)$.
With the dual automata modalities, we can transform every
\emitl{} formula $\phi$ into \emph{negative normal form}, i.e.~an \emitl{} formula using
only atomic propositions, their negations, and the operators $\lor$,
$\land$, $\A_I$, and $\tilde{A}_I$.
It is easy to see that the standard \mitl{} `until' $\phi_1 \until_I \phi_2$
can be defined in terms of automata modalities.
We also use the usual shortcuts like
$\eventually_I\phi \equiv \top\until_I\phi$, 
$\globally_I\phi \equiv \neg\eventually_I\neg\phi$,
and $\phi_1\release_I \phi_2 \equiv \neg\big((\neg\phi_1)\until_I(\neg \phi_2)\big)$.
We say that $\rho$ \emph{satisfies} $\phi$ (written $\rho\models\phi$)
iff $(\rho,1)\models\phi$, and we write
$\sem\phi$ for the timed language of $\phi$, i.e.~the set of all timed words satisfying $\phi$.
$\emitl_{0, \infty}$ is the fragment of \emitl{} where all constraining
intervals $I$ must be lower or upper bounds (e.g.,~$< 3$ or $\geq 5$).
\emph{Extended Event Clock Logic} (\eecl{}) is the fragment of \emitl{}
where $\A_I$ is replaced by a more restricted `event-clock' counterpart:
\begin{itemize}
\item $(\rho,i)\models \oset{\triangleright}{\A}_I(\phi_1, \dots, \phi_n)$ iff 
(i) there is a \emph{minimal} position $j \geq i$ such that $\A$ has an accepting run on $a_i \dots a_j$ where $a_\ell \in \{1, \dots, n\}$
and $(\rho, \ell) \models \phi_{a_\ell}$  for each $\ell$, $i \leq \ell \leq j$; and
(ii) $j$ satisfies $\tau_j-\tau_i\in I$. 
\end{itemize}

\paragraph{Timed automata}
Let $X$ be a finite set of \emph{clocks}
($\R_{\geq 0}$-valued variables).
A \emph{valuation} $v$ for $X$ maps each clock $x \in X$ to a value in $\R_{\geq 0}$.
We denote by $\mathbf{0}$ the valuation that maps every clock to $0$,
and we write the valuation simply as a value in $\R_{\geq 0}$
when $X$ is a singleton.
The set $\Guards(X)$ of \emph{clock constraints} $g$ over $X$ is generated
by $g:= \top\mid g\land g \mid x\bowtie c$ where
${\bowtie}\in \{{\leq},{<},{\geq},{>}\}$, $x\in X$, and $c\in\N_{\geq 0}$.
The satisfaction of a clock constraint $g$ by a valuation $v$ (written $v \models g$) is
defined in the usual way, and we write $\sem{g}$ for the set of valuations $v$ satisfying $g$.
For $t\in\R_{\geq 0}$, we let $v +t$
be the valuation defined by $(v +t)(x) = v (x)+t$ for all $x\in
X$. For $\lambda \subseteq X$, we let $v [\lambda \leftarrow 0]$ be the valuation
defined by $(v[\lambda \leftarrow 0])(x) = 0$ if $x\in \lambda$, and
$(v[\lambda \leftarrow 0])(x) = v (x)$ otherwise.

A \emph{timed automaton} (\ta{}) over $\Sigma$ is a tuple
$\A = \langle \Sigma, S, s_0, X, \transitions, \F \rangle$ where $S$ is a finite set of
locations, $s_0 \in S$ is the initial location, $X$ is a finite set of clocks,
$\transitions \subseteq S \times \Sigma \times \Guards(X) \times
2^X \times S$ is the transition relation,
and $\F=\{F_1,\dots,F_n\}$, with $F_i\subseteq S$ for all $i$, $1\leq i \leq n$,
is the set of sets of final locations.\footnote{We adopt 
\emph{generalised B\"uchi} acceptance  for technical convenience; 
indeed, any \ta{} with a generalised B\"uchi acceptance condition can be converted into a
classical B\"uchi \ta{} via a simple standard construction~\cite{Courcoubetis1992}.}
We say that $\A$ is \emph{deterministic} (a \dta{}) iff for each $s \in S$ and $\sigma \in \Sigma$ and
every pair of transitions $(s, \sigma, g^1, \lambda^1, s^1) \in \transitions$ and $(s, \sigma, g^2, \lambda^2, s^2) \in \transitions$, $g^1 \land g^2$ is not satisfiable.
A \emph{state} of $\A$ is a pair $(s, v)$
of a location $s \in S$ and a valuation $v$ for $X$.
A \emph{run} of $\A$ on a timed word
$(\sigma_1,\tau_1)(\sigma_2,\tau_2)\cdots\in T\Sigma^\omega$ is a
sequence of states $(s_0,v_0)(s_1,v_1)\dots$ where (i)
$v_0=\mathbf{0}$ and (ii) for each $i\geq 0$,
there is a transition $(s_i,\sigma_{i+1},g,\lambda,s_{i+1})$
such that 
$v_i +(\tau_{i+1}-\tau_i)\models g$ (let $\tau_0=0$) and
$v_{i+1} =(v_i +(\tau_{i+1}-\tau_i))[\lambda \leftarrow 0]$. 
A run of $\A$ is \emph{accepting} iff
the set of locations it visits infinitely often contains at least
one location from each $F_i$, $1\leq i\leq n$. A timed word
is \emph{accepted} by $\A$ iff $\A$ has an accepting run on it.
We denote by $\sem{\A}$ the timed language accepted by $\A$.
For two \ta{s} $\A^1 = \langle \Sigma,S^1,s_0^1,X^1,\transitions^1,\F^1 \rangle$ and
$\A^2 = \langle \Sigma,S^2,s_0^2,X^2,\transitions^2,\F^2 \rangle$ over a
common alphabet $\Sigma$, the (synchronous) product 
$\A^1 \times \A^2$ is defined as the \ta{} $\langle \Sigma,S,s_0,X,\transitions,\F \rangle$
where (i) $S=S^1 \times S^2$, $s_0 = (s_0^1, s_0^2)$, and $X = X_1 \cup X_2$;
(ii) $((s^1_1,s^2_1),\sigma,g,\lambda,(s^1_2,s^2_2))\in\transitions$
iff there exists $(s^1_1,\sigma,g^1,\lambda^1,s^1_2)\in\transitions^1$
and $(s^2_1,\sigma,g^2,\lambda^2,s^2_2)\in\transitions^2$ such that
$g=g^1\land g^2$ and $\lambda = \lambda^1 \cup \lambda^2$;
and (iii) let $\F^1=\{F_1^1,\dots,F_n^1\}$, $\F^2=\{F_1^2,\dots,F_m^2\}$, then
$\F = \{F_1^1 \times S^2,\dots,F_n^1\times S^2, S^1\times
F_1^2,\dots,S^1\times F_{m}^2\}$. Note in particular that we have $\sem{\A^1 \times \A^2} = \sem{\A^1} \cap \sem{\A^2}$.

\begin{example}
  \begin{figure}[h]
    \centering \scalebox{.75}{\begin{tikzpicture}[node distance =
        2.5cm] \node[initial left,state](0){};
        \node[state, right of=0](1){};
        \node[state, right of=1, accepting](2){};
        \path
        (0) edge[loopabove,->] (0)
        (0) edge[->] node[above, align=center]{$p \wedge \neg q$ \\ $x := 0$} (1)
        (1) edge[loopabove,->] node[above, align=center]{$\neg q$ \\ $x \leq 1$} (1)
        (1) edge[->] node[above] {$x > 1$} (2);
      \end{tikzpicture}}
    \caption{A \ta{} accepting 
      $\sem{\neg \globally(p\Rightarrow \eventually_{\leq 1}q)}$.}
    \label{fig:TA}
  \end{figure}
Consider the \ta{} over $\Sigma_{\{p,q\}}$ in \figurename~\ref{fig:TA}
(following the usual convention, we omit transition labels when they are $\top$'s and use Boolean formulae over atomic propositions to represent letters,
e.g., here $p \land \neg q$ stands for $\{ \sigma \in \Sigma_{\{p,q\}} \mid p \in \sigma, q \notin \sigma \}$). It non-deterministically pick an event where $p$ holds but $q$ does not hold (thus
$\eventually_{\leq }q$ is not fulfilled immediately) and enforces that $q$ does not hold
in the next time unit. In other words, it accepts $\sem{\eventually\big(p \wedge \globally_{\leq 1}(\neg q)\big)} = \sem{\neg \globally(p\Rightarrow \eventually_{\leq 1}q)}$.
\end{example}

\paragraph{Alternation} One-clock
alternating timed automata (\ocata{s}) extend
one-clock timed automata with the power of
\emph{universal choice}. Intuitively, a transition of an \ocata{}
may spawn several copies of the automaton that run in parallel from the
targets of the transition; a timed word is accepted iff
\emph{all} copies accept it.
Formally, for a set $S$ of
locations, let $\Gamma(S)$ be the set of formulae defined by
\[
  \gamma := \top\mid \bot\mid \gamma_1 \lor \gamma_2\mid \gamma_1 \land \gamma_2
  \mid s \mid x\bowtie c\mid x.\gamma
\] 
where $x$ is the single clock, $c\in\N_{\geq 0}$, ${\bowtie} \in\{{\leq},{<},{\geq},{>}\}$, and
$s \in S$ (the construct $x.$ means ``reset $x$''). For a formula $\gamma \in \Gamma(S)$, let its dual $\overline{\gamma} \in \Gamma(S)$
be the formula obtained by applying
\begin{itemize}
\item $\overline{\top} = \bot$; $\overline{\bot} = \top$;
\item $\overline{\gamma_1 \lor \gamma_2} = \overline{\gamma_1} \land \overline{\gamma_2}$;
		$\overline{\gamma_1 \land \gamma_2} = \overline{\gamma_1} \lor \overline{\gamma_2}$;
\item $\overline{s} = s$; $\overline{x \bowtie c} = \neg (x \bowtie c)$; $\overline{x.\gamma} = x.\overline{\gamma}$.
\end{itemize}
An \ocata{} over $\Sigma$ is a tuple $\A=\langle \Sigma, S, s_0, \atransitions, F \rangle$
where $S$ is a finite set of locations, 
$s_0 \in S$ is the initial location,
$\atransitions \colon S\times \Sigma \to \Gamma(S)$ is the transition
function, and $F \subseteq S$ is the set of final locations.
%
A \emph{state} of $\A$ is a pair $(s,v)$ of a location $s \in S$ and
a valuation $v$ for the single clock $x$.
Given a set of states $M$, a formula $\gamma \in \Gamma(S)$
and a clock valuation $v$, we define
\begin{itemize}
\item $M\models_v \top$; $M\models_v \ell$ iff $(\ell,v)\in M$;
  $M\models_v x\bowtie c$ iff $v\bowtie c$; \mbox{$M\models_v x.\gamma$ iff $M\models_0 \gamma$};
\item $M\models_v \gamma_1\land\gamma_2$ iff $M\models_v \gamma_1$
  and $M\models_v \gamma_2$;
\item $M\models_v \gamma_1\lor\gamma_2$ iff $M\models_v \gamma_1$
  or $M\models_v \gamma_2$.
\end{itemize}
We say that $M$ is a \emph{model} of $\gamma$ with respect
to $v$ iff $M \models_v \gamma$.\footnote{Note that $\models_v$ is \emph{monotonic}: if $M \subseteq M'$ and $M \models_v \gamma$ then $M' \models_v \gamma$.}
A run of $\A$ on a timed word
$(\sigma_1,\tau_1)(\sigma_2,\tau_2)\dots \in T\Sigma^\omega$ is a
rooted directed acyclic graph (DAG) $G=\langle V, \to \rangle$ with vertices of the
form $(s,v,i)\in S\times \R_{\geq 0} \times \N_{\geq 0}$, $(s_0,0,0)$ as the root,
and edges as follows: for every vertex $(s,v,i)$, there is a
model $M$ of the formula $\atransitions(s,\sigma_{i+1})$
with respect to $v+(\tau_{i+1}-\tau_i)$ (again, $\tau_0=0$) such that
there is an edge $(s,v,i) \to (s',v',i+1)$ for every state $(s',v')$ in $M$.
A run $G$ of $\A$ is \emph{accepting} iff
every infinite path in $G$ visits $F$ infinitely often.
A timed word is \emph{accepted} by $\A$
iff $\A$ has an accepting run on it. We denote by $\sem{\A}$ the timed language accepted by $\A$.
For convenience, in the sequel we will regard \nfa{s} as
(untimed) \ocata{s} with finite acceptance conditions and
whose transition functions are simply disjunctions over locations.



%

\begin{example}\label{ex:OCATA}
  \begin{figure}[h]
    \centering \scalebox{.75}{\begin{tikzpicture}[node distance =
        2.5cm] \node[initial left,accepting,state](0){$s_0$};
        \node[state, right of=0](1){$\wedge$};
        \node[state, right of=1](2){$s_1$};
        \path
        (0) edge[loopabove,->] node[above]{$\neg p$} (0)
        (0) edge[-] node[above]{$p \wedge \neg q$} (1)
        (0) edge[loopbelow,->] node[below]{$p \wedge q$} (0)
        (1) edge[->,bend left] node[below] {} (0)
        (1) edge[->] node[above] {$x:=0$} (2)
        (2) edge[->,loopabove] node[above]{$\top$} (2)
        (2) edge[->] node[above]{$x\leq 1, q$} (7,0);
      \end{tikzpicture}}
    \caption{An \ocata{} accepting 
      $\sem{\globally(p\Rightarrow \eventually_{\leq 1}q)}$.}
    \label{fig:OCATA}
  \end{figure}
  \begin{figure}[h]
    \centering \scalebox{1}{\begin{tikzpicture}

		  \node[text width=1.2cm]	(0) {$(s_0, 0, 0)$};
        \node[text width=1.6cm, right = 4mm of 0] (1) {$(s_0, 0.42, 1)$};
        \node[text width=1.7cm, above right = 1mm and 4mm of 1] (2) {$(s_0, 0.42, 2)$};
        \node[text width=1.5cm, below right = 1mm and 4mm of 1] (3) {$(s_1, 0, 2)$};
        \node[text width=1.4cm, right = 4mm of 2] (4) {$(s_0, 0.7, 3)$};
        \node[text width=0.5cm, right = 0mm of 4] (6) {$\dots$};

        \path
        (0) edge[->] (1)
        (1.east) edge[->] (2.west)
        (1.east) edge[->] (3.west)
        (2) edge[->] (4);

      \end{tikzpicture}}
    \caption{A run of the \ocata{} in \figurename~\ref{fig:OCATA} on
      the timed word $(\emptyset,0.42)(\{p\},0.42)(\{q\},0.7)\cdots$.}
    \label{fig:OCATArun}
  \end{figure}
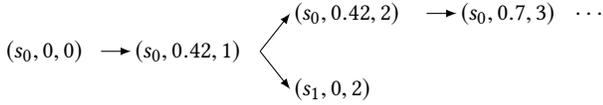
  Consider the \ocata{} over $\Sigma_{\{p,q\}}$ in \figurename~\ref{fig:OCATA}
	which accepts $\sem{\globally(p\Rightarrow \eventually_{\leq 1}q)}$.
	A run of it on 
$(\emptyset,0.42)(\{p\},0.42)(\{q\},0.7)\cdots$
	is depicted in \figurename~\ref{fig:OCATArun} where the root is
  $(s_0,0,0)$. This vertex has a
  single successor $(s_0,0.42,1)$, which in turn has two successors
  $(s_0,0.42,2)$ and $(s_1,0,2)$ (after firing the
  transition $\atransitions(s_0, \{p\}) = s_0 \wedge x.s_1$). Then, $(s_1,0,2)$ has no successor
  since the empty set is a model of $\atransitions(s_1, \{q\}) = x \leq 1$ with respect to $0.28$. 
\end{example}

\paragraph{Verification problems} In this work we are concerned with the following standard
verification problems. Given an \emitl{} formula $\phi$, the \emph{satisfiability} problem
asks whether $\sem{\phi} = \emptyset$. Given a \ta{} $\A$ and
an \emitl{} formula $\phi$, the  \emph{model-checking} problem asks whether
$\sem{\A} \subseteq \sem{\phi}$. As \ta{s} are closed under intersection and
the \emph{emptiness} problem for \ta{s} is decidable, both problems above
can be solved by first translating $\phi$ into an equivalent \ta{} $\A_\phi$.

\section{Expressiveness}

In this section we study the expressiveness of  
$\emitl_{0, \infty}$, \eecl{}, and a `counting' extension of \emitl{}.
It turned out that the class of timed languages captured by
\emitl{} is robust in the sense that it remains the same
under all these modifications.
For the purpose of the proofs below, let us assume (without loss of generality) that
the automaton $\A = \langle \Sigma, S, s_0, \atransitions, F \rangle$ in question is a \dfa{}
and at most one of $\phi_1, \dots, \phi_n$ may hold at any position in a given timed word~\cite{Wolper1994}.

\paragraph{Counting in intervals}
Recall that the constraining intervals $I$ in the
counting modalities in \examplename~\ref{ex:icd}
satisfy $0 \in I$; this non-trivial extension of \mtl{} (and \mitl{}) was first considered by
Hirshfeld and Rabinovich~\cite{HirRab99, Rabinovich2007}.
For the case of timed words, it is shown in~\cite{KriMad16} that allowing arbitrary $I$ (e.g.,~$(1, 2)$)
makes the resulting logic even more expressive. 
Here we show that, by contrast, adding the ability to count in $I$---regardless of whether $0 \in I$---does not increase the expressive power of \emitl{}.\footnote{As \emitl{} can easily express the `until with threshold' modalities of \ctmitl{},
the latter is clearly subsumed by \emitl{}.}
We consider an extention of \emitl{} (which we call~\cemitl{})
that enables specifying the number of positions within a given interval $I$ from now at which final locations can be reached.
More precisely, we have the following semantic clause in \cemitl{}:
\begin{itemize}
\item $(\rho,i)\models \A^{\geq k}_I(\phi_1, \dots, \phi_n)$ iff there exists
  $j_1 < \dots < j_k$ such that for each $\ell$, $1 \leq \ell \leq k$, (i) $j_\ell \geq i$; (ii) $\tau_{j_\ell}-\tau_{i} \in I$; and (iii) there is an
accepting run of $\A$ on some $a_i \dots a_{j_{\ell}}$ where $a_{\ell'} \in \{1, \dots, n\}$
and $(\rho, \ell') \models \phi_{a_{\ell'}}$ for each $\ell'$, $i \leq \ell' \leq j_\ell$.  
\end{itemize}

\begin{figure}[h]
\centering
	\scalebox{.75}{\begin{tikzpicture}[node distance = 2.5cm]
		\node[initial left ,state] (0) {$s_0^1$};
		\node[state, right of=0] (1) {$s_1^1$};
		\node[state, right of=1] (2) {$s_2^1$};
		\path
		(0) edge[loopabove, ->] (0)
		(1) edge[loopabove, ->] (1)
		(0) edge[->] (1)
		(1) edge[->] (2)
		(2) edge[bend left, ->] (0)
		(2) edge[loopabove, ->] (2);
	\end{tikzpicture}}
\caption{$\A^1$ in the proof of Theorem~\ref{thm:countinguseless}.}
\label{fig:counting}
\end{figure}
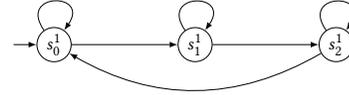

\begin{theorem}\label{thm:countinguseless}
$\cemitl$ and \emitl{} are equally expressive over timed words.
\end{theorem}
\begin{proof}
We give an \emitl{} equivalent of $\A^{\geq k}_I(\phi_1, \dots, \phi_n)$.
Provided that $\phi_1$, \dots, $\phi_n$ are already in \emitl{} and $\A$ is deterministic
in the sense above, we can count modulo $k$ the number of positions
where final locations are reached and ensures that $I$ encompasses all possible values of the counter;
in contrast to~\cite{KriMad16}, here the counter can be implemented directly using automata modalities.
We give a concrete example which should illustrate the idea.
Let $k = 3$ and $\A^2$ be the product of $\A$ and $\A^1$ (\figurename~\ref{fig:counting}),
i.e.~each location of $\A^2$ is of the form $\langle s, s^1 \rangle$ where $s \in S$ and $s^1 \in \{s_0^1, s_1^1, s_2^1 \}$,
and it is accepting iff $s$ and $s^1$ are both final.
Then, let $\A^3$ be the automaton obtained from $\A^2$
by:
\begin{itemize}
\item For all the transitions $\langle s, s_0^1 \rangle \rightarrow \langle s', s_1^1 \rangle$,
$\langle s, s_1^1 \rangle \rightarrow \langle s', s_2^1 \rangle$,
and $\langle s, s_2^1 \rangle \rightarrow \langle s', s_0^1 \rangle$,
keeping only those with $s' \in F$;
\item For all the transitions $\langle s, s_0^1 \rangle \rightarrow \langle s', s_0^1 \rangle$,
$\langle s, s_1^1 \rangle \rightarrow \langle s', s_1^1 \rangle$,
and $\langle s, s_2^1 \rangle \rightarrow \langle s', s_2^1 \rangle$,
keeping only those with $s' \notin F$.
\end{itemize}
Now let $\A^{1,\ell}$ ($\ell \in \{0, 1, 2\}$) be the automaton
obtained from $\A^1$ by adding an extra final location $s^1_F$
and the transition $s^1_{\ell - 1\ (\mathrm{mod}\ 3)} \rightarrow s^1_F$,
and let $\A^{3,\ell}$ be the corresponding product with $\A$,
keeping transitions $\langle s, s_{\ell - 1\ (\mathrm{mod}\ 3)}^1 \rangle \rightarrow \langle s', s_F^1 \rangle$
with $s' \in F$. The original formula $\A_I^{\geq 3}$ is equivalent
to $\bigwedge_{\ell \in \{0, 1, 2\}} \A^{3, \ell}_I$.
\end{proof}

\paragraph{Restricting to event clocks}
We show that the equivalence of \ecl{} and $\mitl_{0, \infty}$
carries over to the current setting. 
More specifically, an \eecl{} formula can be translated
into an equilvalent $\emitl_{0, \infty}$ formula of polynomial size (in DAG representation).
On the other hand, our translation from $\mitl_{0, \infty}$
to \eecl{} induces an exponential blow-up due to the fact that
automata $\A$ have to be determinised.

\begin{theorem}\label{thm:eeclemitl}
\eecl{} and $\emitl_{0, \infty}$ are equally expressive over timed words.
\end{theorem}
\begin{proof}
Again, we assume that the arguments $\phi_1$, \dots, $\phi_n$ are already in the target logic.
The direction from \eecl{} to $\emitl_{0, \infty}$ is simple and almost identical
to the translation from \ecl{} to $\mitl_{0, \infty}$; for example, $\oset{\triangleright}{\A}_{(3, 5)}$
can be written as $\A_{<5} \wedge \neg \A_{\leq 3}$. For the other direction
consider the following $\emitl_{0, \infty}$ formulae:
\begin{itemize}
\item $(\rho, i) \models \A_{\leq c}$: the equivalent formula is simply $\oset{\triangleright}{\A}_{\leq c}$.
\item $(\rho, i) \models \A_{\geq c}$: as in~\cite{Raskin1999}, we consider the subcases where:
\begin{itemize}
\item There is no event in $[\tau_i, \tau_i + c)$ apart from $(\sigma_i, \tau_i)$:
let $\A^2$ be the product of $\A$ and $\A^1$ where $\A^1$ is the automaton depicted in \figurename~\ref{fig:cequals0}.  
We have $(\rho, i) \models \neg \oset[0ex]{\triangleright}{\A}^1_{<c} \wedge \oset[0ex]{\triangleright}{\A}^2_{\geq c}$.

\item There are events in $[\tau_i, \tau_i + c)$ other than $(\sigma_i, \tau_i)$:
let the last event in $[\tau_i, \tau_i + c)$ be $(\sigma_j, \tau_j)$
and $k > j > i$ be the minimal position such that there exists
$a_i \dots a_k \in \sem{\A}$ with $(\rho, \ell) \models \phi_{a_\ell}$ for all $\ell$, $i \leq \ell \leq k$.
By assumption, $a_i \dots a_k$ is unique and $\A$ must reach a  
specific location $s \in S$ after reading $a_i \dots a_j$.
The idea is to split the unique run of $\A$ on $a_i \dots a_k$
at $s$: we take a disjunction over all possible $s \in S$,
enforce that $\tau_j - \tau_i < c$ and $\A$ reaches
a final location from $s$ by reading $a_{j+1} \dots a_k$.
More specifically, let $\B^{s, \phi}$ be the automaton obtained from $\A$
by adding a new location $s_F$, declaring it as the only final location,
and adding new transitions $s' \xrightarrow{\phi_a \wedge \phi} s_F$
for every $s' \xrightarrow{\phi_a} s$ in $\A$.
Let $\C^s$ be the automaton obtained from $\A$ by adding new non-final locations $s_0'$ and $s_1'$,
adding new transitions $s_0' \rightarrow s_1'$
(i.e.~labelled with $\top$) and $s_1' \xrightarrow{\phi_a} s''$
for every $s \xrightarrow{\phi_a} s''$ in $\A$,
removing outgoing transitions from all the final locations, and finally 
setting the initial location to $s_0'$.
We have $(\rho, i) \models \oset[0ex]{\triangleright}{\A}^1_{<c} \wedge \oset{\triangleright}{\A} \wedge \neg \bigvee_{s \in S} \oset[0ex]{\triangleright}{\B}^{s, \phi}_{<c}$
where $\phi = \neg \C^s$. 
\end{itemize}
The equivalent formula is the disjunction of these.
\end{itemize}
The other types of constraing intervals, such as $[0, c)$, are handled almost identically.
\end{proof}

\begin{figure}[h]
\centering
	\scalebox{.75}{\begin{tikzpicture}[node distance = 2.5cm]
		\node[initial left ,state] (0) {$s_0^1$};
		\node[state, right of=0] (1) {$s_1^1$};
		\node[state, accepting, right of=1] (2) {$s_2^1$};
		\path
		(0) edge[loopabove, ->] (0)
		(0) edge[->] (1)
		(1) edge[->] (2)
		(2) edge[loopabove, ->] (2);
	\end{tikzpicture}}
\caption{$\A^1$ in the proof of Theorem~\ref{thm:eeclemitl}.}
\label{fig:cequals0}
\end{figure}
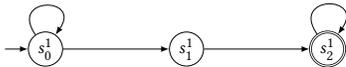

\paragraph{Restricting to one-sided constraining intervals}
Recall that a fundamental stumbling block in the algorithmic analysis of \ta{s}
is that the \emph{universality} problem is undecidable~\cite{AluDil94}.
\dta{s} with finite acceptance conditions, on the other hand, can be complemented easily and
have a decidable universality problem. 
This raises the question of whether one can extend \mitl{} with \dta{} modalities
without losing decidability (both are fully decidable formalisms).
Perhaps surprisingly, the resulting formalism already subsumes \mtl{}
even when punctual constraints are disallowed.
For example, $\eventually_{[d, d]} \phi$ can be written as
$\neg \A' \land \neg \A'' \land \eventually_{[1, \infty)} \phi$
where $\A'$ and $\A''$ are the one-clock deterministic \ta{s}
in \figurename~\ref{fig:punc1} and \figurename~\ref{fig:punc2}, respectively
(in particular, note that $\A'$ and $\A''$ only use lower- and upper-bound constraints).
It follows from~\cite{Ouaknine2006} that the satisfiability problem 
for this formalism is undecidable.
\begin{figure}[h]
\begin{minipage}[b]{0.35\linewidth}
\centering
	\scalebox{.75}{\begin{tikzpicture}[node distance = 2.5cm]
		\node[initial left ,state] (0) {};
		\node[state, accepting, right of=0] (1) {};
		\path
		(0) edge[loopabove, ->] node[above] {$x < d$} (0)
		(0) edge[->] node[above] {$x > d$} (1);
	\end{tikzpicture}}
\caption{$\A'$.}
\label{fig:punc1}
\end{minipage}
\begin{minipage}[b]{0.55\linewidth}
\centering
	\scalebox{.75}{\begin{tikzpicture}[node distance = 2.5cm]
		\node[initial left ,state] (0) {};
		\node[state, right= 2cm of 0] (1) {};
		\node[state, accepting, right= 2cm of 1] (2) {};
		\path
		(0) edge[loopabove, ->] node[above left] {$x < d$} (0)
		(1) edge[loopabove, ->] node[above] {$\neg \phi$, $x \leq d$} (1)
		(0) edge[->] node[above] {$\neg \phi$, $x \geq d$} (1)
		(1) edge[->] node[above] {$x > d$} (2);
	\end{tikzpicture}}
\caption{$\A''$.}
\label{fig:punc2}
\end{minipage}
\end{figure}
Based on a similar trick, we obtain the main result of this section: 
$\emitl_{0, \infty}$ already has the full expressive power of \emitl{}. 
This, together with the fact that the satisfiability and model-checking problems
for $\emitl_{0, \infty}$ are only $\mathrm{PSPACE}$-complete
(Theorem~\ref{thm:pspace}) as compared with $\mathrm{EXPSPACE}$-complete for full $\emitl$~\cite{Ferrere18}, makes $\emitl_{0, \infty}$ a
competitive alternative to other real-time specification formalisms---while a translation from $\emitl$ to $\emitl_{0, \infty}$
inevitably induces at least an exponential blow-up, it can be argued that many properties of practical interest
can be written in $\emitl_{0, \infty}$ directly (e.g.,~\examplename~\ref{ex:icd}
and \examplename~\ref{ex:overtaking}).
The idea of the proof below is similar to that of~\cite[Lemma 6.3.11]{Raskin1999}
($\mitl_{0, \infty}$ and \mitl{} are equally expressive in the continuous semantics),
but the technical details are more involved 
due to automata modalities and the fact that 
each event is not necessarily preceded by another one exactly $1$ time unit earlier
in a timed word;
the latter is essentially the reason why
the expressive equivalence of $\mitl_{0, \infty}$ and \mitl{} fails to hold in the pointwise semantics.

\begin{theorem}\label{thm:expeq}
$\emitl_{0, \infty}$ and \emitl{} are equally expressive over timed words.
\end{theorem}
\begin{proof}
We explain in detail below how to write the \emitl{} formula $\A_{(c, c+1)} (\phi_1, \dots, \phi_n)$ where
$c \geq 0$, and $\phi_1, \dots, \phi_n  \in \emitl_{0, \infty}$
as an $\emitl_{0, \infty}$ formula; the other cases, such as $(c, c+1]$ 
and $[c, c+1]$, are similar.

First consider $c = 0$. If $(\rho, i) \models \A_{(0, 1)}$
for $\rho = (\sigma_1, \tau_1)(\sigma_2,\tau_2)\dots$ and $i \geq 1$, the finite word
$a_i \dots a_k$ accepted by $\A$ must be at least two letters long.
This again is enforced by $\A^1$ in \figurename~\ref{fig:cequals0}:
let $\A^2$ be the product of $\A$ and $\A^1$.
Then, let $\A^3$ be the automaton obtained from $\A^2$
by adding $\neg \nextx_{> 0} \top$ ($\nextx$ is the standard \mitl{} `next' operator~\cite{OuaWor07}) to all the transitions $\langle s, s_0^1 \rangle \rightarrow \langle s', s_0^1 \rangle$ and $\nextx_{> 0} \top$
to all the transitions $\langle s, s_0^1 \rangle \rightarrow \langle s', s_1^1 \rangle$ as conjuncts
(in doing so, extend the alphabet as necessary).
It is not hard to see that $(\rho, i) \models \A^3_{< 1}$
in the two possible situations:
(i) $\tau_{i+1} - \tau_i > 0$ and
(ii) $\tau_j - \tau_i > 0$ for some $j > i + 1$ and $\tau_\ell - \tau_i = 0$ for all $\ell$, $i < \ell < j$.
The other direction ($(\rho, i) \models \A^3_{< 1} \Rightarrow (\rho, i) \models \A_{(0, 1)}$) is straightforward.
It follows that the equivalent $\emitl_{0, \infty}$ formula is $\A^3_{< 1}$.


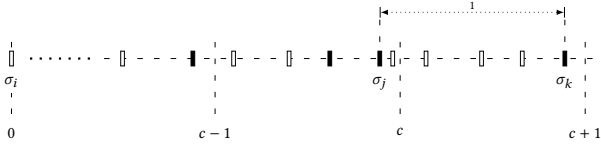
\begin{figure}[h]
\centering
\scalebox{.70}{
\begin{tikzpicture}
\begin{scope}

\draw[-, loosely dashed] (-70pt,0pt) -- (200pt,0pt);


\draw[-, very thick, loosely dotted] (-110pt,0pt) -- (-75pt,0pt);

\draw[loosely dashed] (-120pt,-30pt) -- (-120pt,10pt) node[at start, below=2mm] {$0$};

\draw[loosely dashed] (-10pt,-30pt) -- (-10pt,10pt) node[at start, below=2mm] {$c-1$};

\draw[loosely dashed] (90pt,-30pt) -- (90pt,10pt) node[at start, below=2mm] {$c$};

\draw[loosely dashed] (79pt,0pt) -- (79pt,25pt);
\draw[loosely dashed] (179pt,0pt) -- (179pt,25pt);

\draw[loosely dashed] (190pt,-30pt) -- (190pt,10pt) node[at start, below=2mm] {$c+1$};

\draw[draw=black, fill=white] (-121pt, -4pt) rectangle (-119pt, 4pt);
\node[below, fill=white, inner sep=1mm] at (-120pt, -7pt) {$\sigma_i$};

\draw[draw=black, fill=white] (-61pt, -4pt) rectangle (-59pt, 4pt);

\draw[draw=black, fill=white] (1pt, -4pt) rectangle (-1pt, 4pt);

\draw[draw=black, fill=black] (53pt, -4pt) rectangle (51pt, 4pt);

\draw[draw=black, fill=white] (105pt, -4pt) rectangle (103pt, 4pt);

\draw[draw=black, fill=white] (157pt, -4pt) rectangle (155pt, 4pt);

\draw[draw=black, fill=black] (-21pt, -4pt) rectangle (-23pt, 4pt);

\draw[draw=black, fill=white] (31pt, -4pt) rectangle (29pt, 4pt);

\draw[draw=black, fill=black] (80pt, -4pt) rectangle (78pt, 4pt);
\node[below, fill=white, inner sep=1mm] at (79pt, -7pt) {$\sigma_j$};

\draw[draw=black, fill=white] (87pt, -4pt) rectangle (85pt, 4pt);

\draw[draw=black, fill=white] (135pt, -4pt) rectangle (133pt, 4pt);

\draw[draw=black, fill=black] (180pt, -4pt) rectangle (178pt, 4pt);
\node[below, fill=white, inner sep=1mm] at (179pt, -7pt) {$\sigma_k$};

\end{scope}

\begin{scope}
\draw[|<->|][dotted] (79pt,25pt)  -- (179pt,25pt) node[midway,above] {{\scriptsize $1$}};
\end{scope}

\end{tikzpicture}
}
\caption{Case (i) in the proof of Theorem~\ref{thm:expeq}; solid boxes indicate when $\A$ accepts
the corresponding prefix of $a_i \dots a_k$.}
\label{fig:exactly1before}
\end{figure}

\begin{figure}[h]
\centering
\scalebox{.70}{
\begin{tikzpicture}
\begin{scope}

\draw[-, loosely dashed] (-70pt,0pt) -- (200pt,0pt);


\draw[-, very thick, loosely dotted] (-110pt,0pt) -- (-75pt,0pt);

\draw[loosely dashed] (-120pt,-30pt) -- (-120pt,10pt) node[at start, below=2mm] {$0$};

\draw[loosely dashed] (-10pt,-30pt) -- (-10pt,10pt) node[at start, below=2mm] {$c-1$};

\draw[loosely dashed] (90pt,-30pt) -- (90pt,10pt) node[at start, below=2mm] {$c$};

\draw[loosely dashed] (79pt,0pt) -- (79pt,25pt);
\draw[loosely dashed] (179pt,0pt) -- (179pt,25pt);

\draw[loosely dashed] (190pt,-30pt) -- (190pt,10pt) node[at start, below=2mm] {$c+1$};

\draw[draw=black, fill=white] (-121pt, -4pt) rectangle (-119pt, 4pt);
\node[below, fill=white, inner sep=1mm] at (-120pt, -7pt) {$\sigma_i$};

\draw[draw=black, fill=white] (-61pt, -4pt) rectangle (-59pt, 4pt);

\draw[draw=black, fill=white] (1pt, -4pt) rectangle (-1pt, 4pt);

\draw[draw=black, fill=white] (53pt, -4pt) rectangle (51pt, 4pt);

\draw[draw=black, fill=white] (105pt, -4pt) rectangle (103pt, 4pt);

\draw[draw=black, fill=white] (157pt, -4pt) rectangle (155pt, 4pt);

\draw[draw=black, fill=black] (-21pt, -4pt) rectangle (-23pt, 4pt);

\draw[draw=black, fill=white] (31pt, -4pt) rectangle (29pt, 4pt);


\draw[draw=black, fill=black] (87pt, -4pt) rectangle (85pt, 4pt);
\node[below, fill=white, inner sep=1mm] at (86pt, -7pt) {$\sigma_j$};

\draw[draw=black, fill=white] (135pt, -4pt) rectangle (133pt, 4pt);

\draw[draw=black, fill=black] (180pt, -4pt) rectangle (178pt, 4pt);
\node[below, fill=white, inner sep=1mm] at (179pt, -7pt) {$\sigma_k$};

\end{scope}

\begin{scope}
\draw[|<->|][dotted] (79pt,25pt)  -- (179pt,25pt) node[midway,above] {{\scriptsize $1$}};
\end{scope}

\end{tikzpicture}
}
\caption{Case (ii) in the proof of Theorem~\ref{thm:expeq}.}
\label{fig:blackright}
\end{figure}
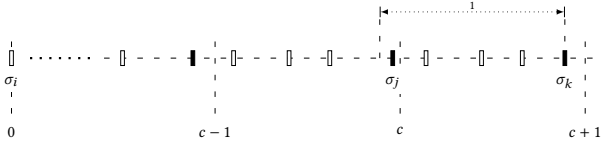

Now consider $c > 0$.
Suppose that $(\rho, i) \models \A_{(c, c + 1)}$
for $\rho = (\sigma_1, \tau_1)(\sigma_2,\tau_2)\dots$ and $i \geq 1$,
let $k > i$ be the \emph{minimal} position such that
$\tau_k - \tau_i \in (c, c+1)$ and there exists $a_i \dots a_k \in \sem{\A}$
with $(\rho, \ell) \models \phi_{a_\ell}$ for all $\ell$, $i \leq \ell \leq k$
(since at most one of $\phi_1$, \dots, $\phi_n$ may hold at any position, we fix
$a_i \dots a_k$ below). 
Consider the following cases (note that they are not mutually disjoint):
\begin{enumerate}[label=(\roman*)]
\item There exists a \emph{maximal} $j$, $i < j < k$ such that (a) $\tau_k - \tau_j = 1$ 
and (b) there is no $\ell$, $j < \ell < k $ such that $a_i \dots a_{\ell} \in \sem{\A}$
with $(\rho, \ell') \models \phi_{a_{\ell'}}$ for all $\ell'$, $i \leq \ell' \leq \ell$
(\figurename~\ref{fig:exactly1before}): 
we take a disjunction over all possible $s \in S$ such that $\A$
reaches $s$ after reading $a_i \dots a_j$ and
enforce that $\tau_j - \tau_i \in (c - 1, c)$ (which we can, by the IH),
 $\A$ reaches a final location from $s$ by reading $a_{j+1} \dots a_k$, and $\tau_k - \tau_j = 1$;
thanks to (b), the last condition, which is otherwise inexpressible in $\emitl_{0, \infty}$,
can be expressed as a conjunction of two formulae labelled with $\leq 1$ and $\geq 1$.
To this end, we use $\B^{s, \phi}$ and $\C^s$ as defined in the proof
of Theorem~\ref{thm:eeclemitl}: we have  
\[
(\rho, i) \models \phi^1 = \bigvee_{s \in S} \B^{s, \phi}_{(c-1, c)} 
\]
where $\phi = \C^s_{\leq 1} \land \C^s_{\geq 1}$.

\item There exists $j$, $i < j < k$ such that $\tau_k - \tau_j < 1$, 
$\tau_j - \tau_i \in (c-1, c]$, and $a_i \dots a_j \in \sem{\A}$
with $(\rho, \ell) \models \phi_{a_{\ell}}$ for all $\ell$, $i \leq \ell \leq j$ (\figurename~\ref{fig:blackright}): let $\D^s$ be the automaton obtained from $\A$
in the same way as $\C^s$ except that we do not remove outgoing transitions
from the final locations.
Regardless of whether there is a event at $\tau_k - 1$, 
it is clear that every position $\ell$ with $\tau_{\ell} - \tau_i \in (c - 1, c]$ must 
satisfy $\C_{(0, 1)}^{s}$ where $s$ is
the location of $\A$ after reading $a_i \dots a_{\ell}$.
We have 
\[
(\rho, i) \models \phi^2 = \A_{(c-1, c]} \land \neg \bigvee_{s \in S} \B^{s, \phi}_{(c-1, c]}
\]
where $\phi = \neg \D^s_{(0, 1)}$.

\item There exists 
$j$, $i < j < k$ such that $\tau_k - \tau_j < 1$, 
$\tau_j - \tau_i \in (c-1, c]$, but 
there is no $\ell$, $i < \ell < k$ such that (a) $\tau_{\ell} - \tau_i \in (c-1, c]$
and (b) $a_i \dots a_{\ell} \in \sem{\A}$
with $(\rho, \ell') \models \phi_{a_{\ell'}}$ for all $\ell'$, $i \leq \ell' \leq \ell$: 
we have 
\[
(\rho, i) \models \phi^3 = \neg \A_{(c-1, c]} \wedge \bigvee_{s \in S} \B^{s, \phi}_{(c-1, c]}
\]
where $\phi = \D^s_{(0, 1)}$.

\item There exists a \emph{maximal} $j$, $i < j < k$ such that $\tau_k - \tau_j > 1$, $\tau_j - \tau_i \in (c-1, c]$,
and there is no $\ell$, $j < \ell < k$ such that (a) $\tau_{\ell} - \tau_i \in (c-1, c]$
and (b) $a_i \dots a_{\ell} \in \sem{\A}$
with $(\rho, \ell') \models \phi_{a_{\ell'}}$ for all $\ell'$, $i \leq \ell' \leq \ell$:
observe that (provided that $s$'s are correctly instantiated to the locations $\A$ reaches
as it reads $a_i \dots a_k$)
while $\C^s_{>1}$ 
may hold arbitrarily often in $[\tau_i, \tau_i + c]$,
the number of positions $\ell$, $i \leq \ell \leq j$ satisfying
\begin{equation}\label{eq:count}
(\rho, \ell) \models \C^s_{>1} \wedge (\rho, \ell + 1) \models \big(\C^s_{\leq 1} \vee (\C^s_{> 1} \wedge s \in F)\big)
\end{equation}
is at most $c$
(since any two of such positions must be separated by more than $1$ time unit).
We define a family of automata modalities $\{\E^m \mid m \geq 1 \}$ such that
each location of $\E^m$ is of the form $\langle s, d \rangle$
with $s \in S$ and $1 \leq d \leq 2m$; see \figurename~\ref{fig:example-b} for an illustration. 
Each transition updates the $s$-component as $\A$ would, enforces the formula labelled on the corresponding transition of $\A$
and, additionally, the formula as labelled in \figurename~\ref{fig:example-b} (with $s$ being the \emph{target} location of
the corresponding transition of $\A$).
The formula $\phi^s$ (illustrated in \figurename~\ref{fig:example-b2}), which also follows $\A$
with an $s$-component, checks that the next position either satisfies (a) $s \in F$, 
or (b) $\C^s_{\leq 1}$ holds continuously until $s \in F$ eventually holds.
Let $\hat{\E}^m$ be obtained from $\E^m$ by `inlining' $\phi^s$: removing the leftmost locations of 
$\phi^s$ and merge the middle locations of $\phi^s$ with the rightmost locations of $\E^m$.
Apparently, $\hat{\E}^m$ and ${\E^m}$ are equivalent if there is no constraing interval---the only
difference between them is which position is `timed'.
Now suppose that the number of positions $\ell$, $i \leq \ell \leq j$ satisfying (\ref{eq:count}) is $m$. 
Since $j$ is the last of these positions, we have $(\rho, i) \models \hat{\E}^m_{< c + 1}$.
On the other hand, as there are only $m - 1$ such positions in $[\tau_i, \tau_i + c - 1]$,
we have $(\rho, i) \models \neg {\E^{m}}_{\leq c - 1}$.
By the above, we have
\[
(\rho, i) \models \phi^4 = \bigvee_{1 \leq m \leq c} \big(\hat{\E}^m_{< c + 1} \land \neg {\E^{m}}_{\leq c - 1} \big) \,.
\]
\begin{figure}
\centering
	\scalebox{.75}{\begin{tikzpicture}[node distance = 2cm]
		\node[initial left ,state] (0) {};
		\node[state, right= 1.5cm of 0] (1) {};
		\node[state, right= 1.5cm of 1] (2) {};
		\node[state, right= 1.5cm of 2] (3) {};
		\node[state, right= 1.5cm of 3] (4) {};
		\node[state, accepting, right= 1.5cm of 4] (5) {};
		\path
		(0) edge[loopabove, ->] (0)
		(2) edge[loopabove, ->] (2)
		(4) edge[loopabove, ->] (4)
		(0) edge[->] node[above] {$\C^s_{>1}$} (1)
		(1) edge[->] node[above] {$\C^s_{\leq 1}$} (2)
		(1) edge[->, bend right=30] node[below] {$\C^s_{>1} \land s \in F$} (3)
		(2) edge[->] node[above] {$\C^s_{>1}$} (3)
		(3) edge[->] node[above] {$\C^s_{\leq 1}$} (4)
		(3) edge[->, bend right=30] node[below] {$\C^s_{>1} \land s \in F \land \phi^s$} (5)
		(4) edge[->] node[above] {$\C^s_{>1} \land \phi^s$} (5)

		;
	\end{tikzpicture}}
  \caption{An illustration of $\E^3$ in the proof of Theorem~\ref{thm:expeq}.}
  \label{fig:example-b}
\end{figure}
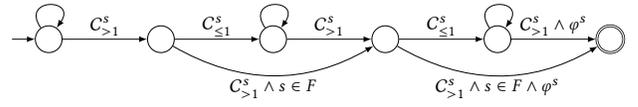

\begin{figure}
\centering
	\scalebox{.75}{\begin{tikzpicture}[node distance = 2cm]
		\node[initial left ,state] (0) {};
		\node[state, right= 1.5cm of 0] (1) {};
		\node[state, accepting, right= 1.5cm of 1] (2) {};
		\path
		(1) edge[loopabove, ->] node[above] {$\C^s_{\leq 1}$} (1)
		(0) edge[->] (1)
		(1) edge[->] node[above] {$s \in F$} (2)

		;
	\end{tikzpicture}}
  \caption{An illustration of $\phi^s$ in the proof of Theorem~\ref{thm:expeq}.}
  \label{fig:example-b2}
\end{figure}

\item There is no event in $(\tau_i + c -1, \tau_i + c]$: We have 
\[
(\rho, i) \models \phi^5 = \neg \eventually_{(c-1, c]} \top \land \phi' \,.
\]
If $c = 1$ then $\phi'$ can simply be taken as $\A_{(0, 2)}$,
which is equivalent to $\A^3_{< 2}$ by the same argument as before.
If $c > 1$, then $\phi'$ can be taken as
\[
\bigvee_{1 \leq m \leq c - 1} \big(\hat{\E}^m_{< c + 1} \land \neg {\E^{m}}_{\leq c - 2} \big) \vee \phi''
\]
where $\phi''$ is
\[
\neg \eventually_{(c-2, c-1]} \top \wedge \bigvee_{1 \leq m \leq c - 2} \big(\hat{\E}^m_{< c + 1} \land \neg {\E^{m}}_{\leq c - 3} \big) \vee \phi''' \,.
\]
Intuitively, the former part of $\phi'$ is used to handle the case when
there is (at least) a event in $(\tau_i + c-2, \tau_i + c-1]$, and
the former part of $\phi''$ is for $(\tau_i + c-3, \tau_i + c-2]$,
and so on.

\end{enumerate}
We omit the other direction as it is (more or less) straightforward.
The equivalent $\emitl_{0, \infty}$ formula is $\phi^1 \vee \phi^2 \vee \phi^3 \vee \phi^4 \vee \phi^5$.
\end{proof}

\section{From $\emitl_{0, \infty}$ to timed automata}\label{sec:emitl2ocata}

\paragraph{Embedding \emitl{} formulae into \ocata{s}}
We give a translation from a given \emitl{} formula $\phi$ over $\AP$ (which we assume to be in negative normal form)
into an \ocata{} $\A_\phi = \langle \Sigma_\AP, S, s_0, \atransitions, F \rangle$ 
such that $\sem{\A_\phi} = \sem{\phi}$.
While this mostly follows the lines of the translation for \mtl{} (and \mitl{}) in~\cite{OuaWor07, BriEst14},
it is worth noting that the 
resulting \ocata{} $\A_\phi$ is \emph{weak}~\cite{MulSao86,KupVar97}
but not necessarily \emph{very-weak}~\cite{Rohde97phd, GasOdd01} due to the presence of automata modalities.
The set of locations $S$ of $\A_\phi$
contains (i) $s^\textit{init}$; 
(ii) all the locations of $\A$ for every subformula $\A_I(\phi_1, \dots, \phi_n)$;
(iii) all the locations of $\A$ for every subformula $\tilde{\A}_I(\phi_1, \dots, \phi_n)$. 
The initial location $s_0$ is $s^\textit{init}$, and the final locations $F$ are all the locations
of $\A$ for every subformula $\tilde{\A}_I(\phi_1, \dots, \phi_n)$. Finally, 
for each $\sigma \in \Sigma_{\AP}$, $\atransitions$ is defined inductively as follows (let $\A = \langle \Sigma^\A, S^\A, s_0^\A, \atransitions^\A, F^\A \rangle$ with $\Sigma^\A = \{1, \dots, n\}$):
\begin{itemize}
\item
  $\atransitions(s^\textit{init},\sigma)=x.\atransitions(\phi,\sigma)$,
  $\atransitions(\top,\sigma)=\top$, and
  $\atransitions(\bot,\sigma)=\bot$;
\item $\atransitions(p,\sigma) = \top$ if $p\in\sigma$,
  $\atransitions(p,\sigma)=\bot$ otherwise;
\item $\atransitions(\lnot p,\sigma) = \top$ if $p\notin \sigma$,
  $\atransitions(\lnot p,\sigma)=\bot$ otherwise;
\item
  $\atransitions(\phi_1\lor\phi_2,\sigma)=\atransitions(\phi_1,\sigma)\lor
  \atransitions(\phi_2,\sigma)$, and
  $\atransitions(\phi_1\land\phi_2,\sigma)=\atransitions(\phi_1,\sigma)\land \atransitions(\phi_2,\sigma)$;
\item
  $\atransitions(\A_I(\phi_1, \dots, \phi_n),\sigma)= x.\atransitions(s_0^\A, \sigma)$;
\item
  $\atransitions(s^\A,\sigma)= \bigvee_{a \in \Sigma^\A} \big( \atransitions(\phi_a , \sigma) \land \atransitions^\A[s_F^\A \leftarrow s^\A_F \lor x \in I](s^\A, a) \big)$ where $s^\A \in S^\A$ and $\atransitions^\A[s_F^\A \leftarrow s^\A_F \lor x \in I]$ is obtained from $\atransitions^\A$ by substituting every $s_F^\A \in F^\A$ with
$s_F^\A \lor x \in I$ for some subformula $\A_I(\phi_1, \dots, \phi_n)$;
\item
  $\atransitions(\tilde{\A}_I(\phi_1, \dots, \phi_n),\sigma)= x.\atransitions(s_0^\A, \sigma)$;
\item
  $\atransitions(s^\A,\sigma)= \bigwedge_{a \in \Sigma^\A} \big( \atransitions(\phi_a , \sigma) \lor \overline{\atransitions^\A}[s_F^\A \leftarrow s^\A_F \land x \notin I](s^\A, a) \big)$ where $s^\A \in S^\A$ and $\overline{\atransitions^\A}[s_F^\A \leftarrow s^\A_F \land x \notin I]$ is obtained from $\overline{\atransitions^\A}$ by substituting every $s_F^\A \in F^\A$ with
$s_F^\A \land x \notin I$ for some subformula $\tilde{\A}_I(\phi_1, \dots, \phi_n)$.
\end{itemize}

\begin{proposition}\label{prop:emitl2ocata}
Given an \emitl{} formula $\phi$ in negative normal form, \mbox{$\sem{\A_\phi} = \sem{\phi}$}.
\end{proposition}

 
We now focus on the case where $\phi$ is an $\emitl_{0, \infty}$ formula
and give a set of component \ta{s} whose product  
`implements' the corresponding \ocata{} $\A_\phi$.
As we will need some notions from~\cite{BriGee17},
we briefly recall them here to keep the paper self-contained.
%

\paragraph{Compositional removal of alternation in $\phi$}

Let $\Phi$ be the set of temporal subformulae (i.e.~whose outermost
operator is $\A_I$ or $\tilde{A}_I$) of $\phi$.  We introduce
a new atomic proposition $p_{\psi}$ for each $\psi\in \Phi$
(the \emph{trigger} for $\psi$)
and let $\AP_\Phi = \{ p_\psi \mid \psi \in \Phi\}$.
For a timed word $\rho'$ over $\Sigma_{\AP \cup \AP_\Phi}$,
we denote by $\proj_\AP(\rho')$ the timed word obtained from $\rho'$
by hiding all $p \notin \AP$ (i.e.~$p \in \AP_\Phi$).
For a timed language $\Lang$ over $\AP \cup \AP_\Phi$
we write $\proj_\AP(\Lang) = \{ \proj_\AP(\rho') \mid \rho' \in \Lang \}$.
Let $\overline{\psi}$ be the formula obtained from an $\emitl_{0, \infty}$ formula $\psi$ (in negative normal form) by replacing all of its 
top-level temporal subformulae by their corresponding triggers, 
i.e.~$\overline{\psi}$ is defined
inductively as follows (where $p\in \AP$):
\begin{itemize}
\item $\overline {\psi_1\land \psi_2} =\overline {\psi_1}\land \overline{\psi_2}$;

\item $\overline {\psi_1\lor \psi_2} =\overline {\psi_1}\lor \overline{\psi_2}$;
\item $\overline{\psi}= \psi\text{ when }\psi\text{ is }\top\text{ or } \bot\text{ or } p\text{ or } \neg p$;
\item $\overline{\psi}= p_\psi \text{ when }\psi\text{ is } \A_I(\phi_1, \dots, \phi_n) \text{ or } \tilde{\A}_I(\phi_1, \dots, \phi_n)$.
\end{itemize}
Note that $\overline{\psi}$ is simply a positive Boolean combination of atomic propositions.
In this way, we can turn the given $\emitl_{0, \infty}$ formula $\phi$ into
an equisatisfiable $\emitl_{0, \infty}$ formula $\phi'$ over $\AP \cup \AP_\Phi$: 
the conjunction of $\overline{\phi}$, 
\[
\bigwedge_{\{ \psi \in \Phi \mid \psi = \A_I(\phi_1, \dots, \phi_n) \}}\globally \big(p_\psi \Rightarrow \A_I(\overline{\phi_1}, \dots, \overline{\phi_n})\big) \,,
\]   
and the counterparts for $\{ \psi \in \Phi \mid \psi = \tilde{\A}_I(\phi_1, \dots, \phi_n) \}$.
Finally, we construct the component \ta{s} $\C^\textit{init}$ (which accepts $\sem{\overline{\phi}}$)
and $\C^\psi$ (which accepts, say, $\sem{\globally \big(p_\psi \Rightarrow \A_I(\overline{\phi_1}, \dots, \overline{\phi_n})\big)}$) for every $\psi \in \Phi$.
The timed language of $\phi'$ is accepted by the product
$\C^\textit{init}\times \prod_{\psi\in\Phi}\C^\psi$ and, in particular, $\proj_\AP(\sem{\C^{\textit{init}}\times \prod_{\psi\in\Phi}\C^\psi}) = \sem{\A_\phi}$.
Intuitively, $p_\psi$ being $\top$ (the trigger $p_\psi$ is `\emph{pulled}')
at some position means that the \ocata{} $\A_\phi$ spawns a copy (several copies)
of $\A$ where $\psi = \A_I$ ($\psi = \tilde{\A}_I$) at this position or,
equivalently, an obligation that $\A_I$ ($\tilde{\A}_I$) must hold
is imposed on this position.

\paragraph{Component \ta{s} for automata modalities}
As the construction of $\C^\textit{init}$ is trivial,
we only describe the component \ta{s} $\C^\psi$ for $\psi = \A_I(\phi_1, \dots, \phi_n)$ and $\psi = \tilde{\A}_I(\phi_1, \dots, \phi_n)$
where $I = [0, c]$ or $I = [c, \infty)$ for some $c \in \N_{\geq 0}$;
the other types of constraining intervals are handled similarly. 
The crucial observation that allows us to bound the number of clocks needed
in $\C^\psi$ is that two or more obligations, provided that
their corresponding copies
of $\A$ are in the same location(s) at some point(s)
and $I$ is a lower or upper bound,
can be merged into a single one.
Instead of keeping track of the order of the values of its clocks,
$\C^\psi$ non-deterministically guesses how obligations
should be merged and put them into suitable sub-components accordingly.
To ensure that all obligations are satisfied, we use an extra
variable $\ell$ such that $\ell = 0$ when there is no obligation,
$\ell = 1$ when there is at least one pending obligation,
$\ell = 2$ when the pending obligations have just been
satisfied and a new obligation has just arrived, and finally
$\ell = 3$ when we have to wait the current obligations to be satisfied
(explained below).
In all the cases below we fix $\Sigma = \Sigma_{\AP \cup \AP_\Phi}$ and $|S^\A| = m$.
We write $S^\textit{src} \xrightarrow[\vee]{\sigma} S^\textit{tgt}$,
where $S^\textit{src}$ and $S^\textit{tgt}$ are two subsets of $S^\A$, iff
$S^\textit{tgt}$ is a minimal set such that
for each $s^{\A, 1} \in S^\textit{src}$, there is a transition
$s^{\A, 1} \xrightarrow{\overline{\phi_a}} s^{\A, 2}$ (where $a \in \{1, \dots, n\}$)
of $\A$
with $\sigma \models \overline{\phi_a}$ and $s^{\A, 2} \in S^\textit{tgt}$.
Similarly, we write $S^\textit{src} \xrightarrow[\wedge]{\sigma} S^\textit{tgt}$ iff
$S^\textit{tgt}$ is a minimal set such that
for each $s^{\A, 1} \in S^\textit{src}$ and each transition
$s^{\A, 1} \xrightarrow{\overline{\phi_a}} s^{\A, 2}$ (where $a \in \{1, \dots, n\}$) of $\A$, either $s^{\A, 2} \in S^\textit{tgt}$ or $\sigma \models \overline{\phi_a}$.

\paragraph{$\psi = \A_{\leq c}$}
Let $C^\psi = \langle \Sigma, S, s_0, X, \transitions, \F \rangle$ be defined as follows
(to simplify the presentation, in this case we assume that $\A$ only accepts words of length $\geq 2$):
\begin{itemize}
\item Each location $s \in S$ is of the form $\langle \ell_1, S_1, \dots, \ell_m, S_m \rangle$ where $\ell_j \in \{0, 1, 2\}$ and $S_j \subseteq S^\A$ for all $j \in \{1, \dots, m\}$; intuitively, $\langle \ell_j, S_j \rangle$ can be seen as a location
of the sub-component $C^\psi_j$; 
\item $s_0 = \langle 0, \emptyset, \dots, 0, \emptyset \rangle$;
\item $X = \{x_1, \dots, x_m\}$;
\item $\F = \{F_1, \dots, F_m\}$ where $F_j$ contains all locations with $\ell_j = 0$ or $\ell_j = 2$;
\item $\transitions$ is obtained by synchronising
the transitions
$\langle \ell, S \rangle \xrightarrow{\sigma, g, \lambda} \langle \ell', S' \rangle$
of individual sub-components (we omit the subscripts for brevity):
	\begin{itemize}
	\item $p_\psi \notin \sigma$; $\ell' = 0$, $\ell = 0$; $S' = \emptyset$, $S = \emptyset$; $g = \top$; $\lambda = \emptyset$.
	\item $p_\psi \notin \sigma$; $\ell' = 1$, $\ell \in \{1, 2\}$; $S \xrightarrow[\vee]{\sigma} S'$; $g = \top$; $\lambda = \emptyset$.
	\item $p_\psi \notin \sigma$; $\ell' = 0$, $\ell \in \{1, 2\}$; $S' = \emptyset$, $S = \{s^\A\}$, $s^\A_F \models \atransitions(s^\A, \sigma)$ for some $s^\A_F \in F^\A$; $g = x \leq c$; $\lambda = \emptyset$.
	\item $p_\psi \in \sigma$; $\ell' = 1$, $\ell = 0$; $\{s_0^\A\} \xrightarrow[\vee]{\sigma} S'$, $S = \emptyset$; $g = \top$; $\lambda = \{ x \}$. 
	\item $p_\psi \in \sigma$; $\ell' = 1$, $\ell \in \{1, 2\}$; $S'$ is the union of
some $S''$ such that $S \xrightarrow[\vee]{\sigma} S''$ and $S'''$ with $\{s_0^\A\} \xrightarrow[\vee]{\sigma} S'''$; $g = \top$; $\lambda = \emptyset$. 
	\item $p_\psi \in \sigma$; $\ell' = 2$, $\ell \in \{1, 2\}$; $\{s_0^\A\} \xrightarrow[\vee]{\sigma} S'$, $S = \{s^\A\}$, $s^\A_F \models \atransitions(s^\A, \sigma)$ for some $s^\A_F \in F^\A$; $g = x \leq c$; $\lambda = \{x \}$. 
	\end{itemize}
If $p_\psi \in \sigma$, then exactly
one of the sub-components takes a `$p_\psi \in \sigma$' transition
while the others proceed as if $p_\psi \notin \sigma$.
\end{itemize}

\begin{proposition}\label{prop:firstcase}
$\sem{\C^\psi} = \sem{\globally \big(p_\psi \Rightarrow \A_{\leq c}(\overline{\phi_1}, \dots, \overline{\phi_n})\big)}$.
\end{proposition}
\begin{proof}[Proof sketch.]
Let $\psi' =  \globally \big(p_\psi \Rightarrow \A_{\leq c}(\overline{\phi_1}, \dots, \overline{\phi_n})\big)$
and $\A_{\psi'}$ be the equivalent \ocata{} obtained via Proposition~\ref{prop:emitl2ocata}.
If a timed word $\rho = (\sigma_i, \tau_i)_{i \geq 1}$ satisfies $\psi'$, there must be an accepting run
$G = \langle V, \rightarrow \rangle$ of $\A_{\psi'}$ on $\rho$;
in particular, a copy of $\A$ is spawned whenever $p_\psi$ holds.
Now consider each `level' 
$L_i = \{ (s, v) \mid (s, v, i) \in V \}$
 of $G$ in the increasing order of $i$.
If $|\{ (s^\A, v) \mid (s^\A, v) \in L_i \}| \leq 1$ for every $s^\A \in S^\A$, 
in $\C^\psi$ we simply put each corresponding obligation into an unused sub-component (with $\ell = 0$ and $S = \emptyset$)
when it arrives, i.e.~$p_\psi$ holds.
If $|\{ (s^\A, v) \mid (s^\A, v) \in L_i \}| > 1$ for some $s^\A \in S^\A$, since
the constraining interval $[0, c]$ is \emph{downward closed}, the DAG
obtained from $G$ by replacing all the subtrees rooted at nodes $(s^\A, v, i)$ 
with $(s^\A, v^\textit{max}, i)$,
where $v^\textit{max} = \max \{ v \mid (s^\A, v) \in L_i \}$,
is still an accepting run of $\A_{\psi'}$; in $\C^\psi$, this amounts to putting
the obligations that correspond to nodes $(s^\A, v, i)$ into the sub-component
that holds the (oldest) obligation that corresponds to $(s^\A, v^\textit{max}, i)$.
We do the same for all such $s^\A$, obtain $G'$, and start over from $i + 1$.
In this way, we can readily construct an accepting run of $\C^\psi$ on $\rho$.
The other direction obviously holds as each sub-component $\C^\psi_j$
does not reset its associated clock $x_j$ when $p_\psi \in \sigma$ and $\ell_j \in \{1, 2\}$,
unless the (only remaining) obligation in $S_j$ is fulfilled right away.
In other words, $\C^\psi_j$ adds an obligation that is at least as strong to $S_j$
without weakening the existing ones in $S_j$.
\end{proof}

\paragraph{$\psi = \A_{\geq c}$}
Let $\C^\psi = \langle \Sigma, S, s_0, X, \transitions, \F \rangle$ be defined as follows:
\begin{itemize}
\item Each location $s \in S$ is of the form $\langle \ell_1, S_1, T_1 \dots, \ell_m, S_m, T_m \rangle$ where $\ell_j \in \{0, 1, 2, 3\}$ and $S_j, T_j \subseteq S^\A$ for all $j \in \{1, \dots, m\}$; intuitively, $\langle \ell_j, S_j, T_j \rangle$ can be seen as a location
of the sub-component $C^\psi_j$;
\item $s_0 = \langle 0, \emptyset, \emptyset, \dots, 0, \emptyset, \emptyset \rangle$;
\item $X = \{x_1, \dots, x_m\}$;
\item $\F = \{F_1, \dots, F_m\}$ where $F_j$ contains all locations with $\ell_j = 0$ or $\ell_j = 2$;
\item $\transitions$ is obtained by synchronising (in the same way as before) 
$\langle \ell, S, T \rangle \xrightarrow{\sigma, g, \lambda} \langle \ell', S', T' \rangle$
of individual sub-components:
	\begin{itemize}
	\item $p_\psi \notin \sigma$; $\ell' = 0$, $\ell = 0$; $S' = \emptyset$, $S = \emptyset$, $T' = \emptyset$, $T = \emptyset$; $g = \top$; $\lambda = \emptyset$.
	\item $p_\psi \notin \sigma$; $\ell' = 1$, $\ell \in \{1, 2\}$; $S \xrightarrow[\vee]{\sigma} S'$, $T' = \emptyset$, $T = \emptyset$; $g = \top$; $\lambda = \emptyset$.
	\item $p_\psi \notin \sigma$; $\ell' = 3$, $\ell = 3$; $S \xrightarrow[\vee]{\sigma} S'$, $T \xrightarrow[\vee]{\sigma} T'$; $g = \top$; $\lambda = \emptyset$.
	\item $p_\psi \notin \sigma$; $\ell' = 0$, $\ell \in \{1, 2\}$; $S' = \emptyset$, $S = \{s^\A\}$, $s^\A_F \models \atransitions(s^\A, \sigma)$ for some $s^\A_F \in F^\A$, $T' = \emptyset$, $T = \emptyset$; $g = x \geq c$; $\lambda = \emptyset$.
	\item $p_\psi \notin \sigma$; $\ell' = 2$, $\ell = 3$; $T \xrightarrow[\vee]{\sigma} S'$, $S = \{s^\A\}$, $s^\A_F \models \atransitions(s^\A, \sigma)$ for some $s^\A_F \in F^\A$, $T' = \emptyset$; $g = x \geq c$; $\lambda = \{ x \}$.
	\item $p_\psi \in \sigma$; $\ell' = 1$, $\ell = 0$; $\{s_0^\A\} \xrightarrow[\vee]{\sigma} S'$, $S = \emptyset$, $T' = \emptyset$, $T = \emptyset$; $g = \top$; $\lambda = \{ x \}$. 
	\item $p_\psi \in \sigma$; $\ell' = 1$, $\ell \in \{1, 2\}$; $S'$ is the union of
some $S''$ such that $S \xrightarrow[\vee]{\sigma} S''$ and $S'''$ with $\{s_0^\A\} \xrightarrow[\vee]{\sigma} S'''$, $T' = \emptyset$, $T = \emptyset$; $g = \top$; $\lambda = \{ x \}$. 
	\item $p_\psi \in \sigma$; $\ell' = 3$, $\ell = 1$; $S \xrightarrow[\vee]{\sigma} S'$, $\{s_0^\A\} \xrightarrow[\vee]{\sigma} T'$, $T = \emptyset$; $g = \top$; $\lambda = \emptyset$. 
	\item $p_\psi \in \sigma$; $\ell' = 3$, $\ell = 3$; $S \xrightarrow[\vee]{\sigma} S'$, $T'$ is the union of some $T''$
such that $T \xrightarrow[\vee]{\sigma} T''$ and $T'''$ with $\{s_0^\A\} \xrightarrow[\vee]{\sigma} T'''$; $g = \top$; $\lambda = \emptyset$. 
	\item $p_\psi \in \sigma$; $\ell' = 2$, $\ell \in \{1, 2\}$; $\{s_0^\A\} \xrightarrow[\vee]{\sigma} S'$, $S = \{s^\A\}$, $s^\A_F \models \atransitions(s^\A, \sigma)$ for some $s^\A_F \in F^\A$, $T' = \emptyset$, $T = \emptyset$; $g = x \geq c$; $\lambda = \{x \}$. 
	\item $p_\psi \in \sigma$; $\ell' = 2$, $\ell = 3$; $S'$ is the union of some $S''$ such that $T \xrightarrow[\vee]{\sigma} S''$ and $S'''$ with $\{s_0^\A\} \xrightarrow[\vee]{\sigma} S'''$, $S = \{s^\A\}$, $s^\A_F \models \atransitions(s^\A, \sigma)$ for some $s^\A_F \in F^\A$, $T' = \emptyset$; $g = x \geq c$; $\lambda = \{ x \}$. 
	\end{itemize}
\end{itemize}
\begin{proposition}
$\sem{\C^\psi} = \sem{\globally \big(p_\psi \Rightarrow \A_{\geq c}(\overline{\phi_1}, \dots, \overline{\phi_n})\big)}$.
\end{proposition}
\begin{proof}[Proof sketch.]
Similar to the proof of Proposition~\ref{prop:firstcase}, but since
$[c, \infty)$ is \emph{upward closed}, we 
replace all the subtrees rooted at nodes $(s^\A, v, i)$ 
with $(s^\A, v^\textit{min}, i)$,
where $v^\textit{min} = \min \{ v \mid (s^\A, v) \in L_i \}$;
in $\C^\psi$, we still put the obligations that correspond to nodes $(s^\A, v, i)$ into the sub-component
that holds the (oldest) obligation that corresponds to $(s^\A, v^\textit{max}, i)$.
There is, however, a potential issue: since we reset $x_j$ whenever the trigger $p_\psi$ is pulled
and $\C^\psi_j$ is chosen, 
it might be the case that $x_j$ never reaches $c$, i.e.~the satisfaction of the obligations in $S_j$ are delayed indefinitely.
Following~\cite{BriGee17}, we solve this by locations with $\ell_j = 3$ such that, when entered, 
we stop resetting $x_j$ and put the new obligations into $T_j$ instead;
when the obligations in $S_j$ are fulfilled, we move the obligations in $T_j$ to $S_j$
and reset $x_j$.
The other direction obviously holds as each sub-component $\C^\psi_j$
resets $x_j$ when $p_\psi \in \sigma$ and $\ell_j \in \{1, 2\}$,
unless it goes from $\ell_j = 1$ to $\ell_j = 3$.
In other words, $\C^\psi_j$ adds the new obligation to $S_j$
while strengthening the existing ones in $S_j$.
\end{proof}

\paragraph{$\psi = \tilde{\A}_{\leq c}$}
Let $\C^\psi = \langle \Sigma, S, s_0, X, \transitions, \F \rangle$ be defined as follows:
\begin{itemize}
\item Each location $s \in S$ is of the form $\langle S_1, \dots, S_{m+1} \rangle$ where $S_j \subseteq S^\A$
		for all $j \in \{1, \dots, m+1\}$; intuitively, $S_j$ can be seen as a location of the sub-component $C^\psi_j$;
\item $s_0 = \langle  \emptyset, \dots, \emptyset \rangle$;
\item $X = \{x_1, \dots, x_{m+1}\}$;
\item $\F = \emptyset$, i.e.~any run is accepting;
\item $\transitions$ is obtained by synchronising (in the same way as before) 
$S \xrightarrow{\sigma, g, \lambda} S'$
of individual sub-components:
	\begin{itemize}
	\item $p_\psi \notin \sigma$; $S' = \emptyset$, $S = \emptyset$; $g = \top$; $\lambda = \emptyset$.

	\item $p_\psi \notin \sigma$; $S \xrightarrow[\wedge]{\sigma} S'$, $S' \cap F^\A = \emptyset$; $g = x \leq c$; $\lambda = \emptyset$.

	\item $p_\psi \notin \sigma$; $S' = \emptyset$; $g = x > c$; $\lambda = \emptyset$.

	\item $p_\psi \in \sigma$; $\{s_0^\A\} \xrightarrow[\wedge]{\sigma} S'$, $S' \cap F^\A = \emptyset$, $S = \emptyset$; $g = \top$; $\lambda = \{ x \}$. 

	\item $p_\psi \in \sigma$; $S'$ is the union of
some $S''$ such that $S \xrightarrow[\wedge]{\sigma} S''$ and $S'''$ with $\{s_0^\A\} \xrightarrow[\wedge]{\sigma} S'''$, $S' \cap F^\A = \emptyset$; $g = x \leq c$; $\lambda = \{ x \}$. 

	\item $p_\psi \in \sigma$; $\{s_0^\A\} \xrightarrow[\wedge]{\sigma} S'$, $S' \cap F^\A = \emptyset$; $g = x > c$; $\lambda = \{x \}$. 

	\end{itemize}
\end{itemize}

\begin{proposition}\label{prop:thirdcase}
$\sem{\C^\psi} = \sem{\globally \big(p_\psi \Rightarrow \tilde{\A}_{\leq c}(\overline{\phi_1}, \dots, \overline{\phi_n})\big)}$.
\end{proposition}
\begin{proof}[Proof sketch.]
Let $\psi' =  \globally \big(p_\psi \Rightarrow \tilde{\A}_{\leq c}(\overline{\phi_1}, \dots, \overline{\phi_n})\big)$
and $\A_{\psi'}$ be the equivalent \ocata{} obtained via Proposition~\ref{prop:emitl2ocata}.
Consider each level $L_i = \{ (s, v) \mid (s, v, i) \in V \}$
 of an accepting run $G = \langle V, \rightarrow \rangle$ of $\A_{\psi'}$ on $\rho = (\sigma_i, \tau_i)_{i \geq 1}$
in the increasing order of $i$. In $\C^\psi$, whenever the trigger $p_\psi$ is pulled, 
we attempt to put the corresponding obligation into an unused sub-component (with $S = \emptyset$)
or a sub-component that can be cleared (with $x > c$);
if this is not possible, since for every $s^\A \in S^\A$ all the subtrees rooted at nodes
$(s^\A, v, i)$ can be replaced with the subtree rooted at 
$(s^\A, v^\textit{min}, i)$ where $v^\textit{min} = \min \{ v \mid (s^\A, v) \in L_i \}$,
at least one sub-component $C^\psi_j$ becomes redundant, i.e.~all of its obligations are implied by
the other sub-components $C^\psi_k$, $k \neq j$. A consequence
is that the obligations in the sub-component $C^\psi_k$ with the minimal non-negative
value of $x_j - x_k$ can be merged with the obligations in $C^\psi_j$, freeing up a sub-component for
the current incoming obligation.
This can be repeated to construct an accepting run of $\C^\psi$ on $\rho$.
The other direction holds as each $C^\psi_j$ adds the new obligation to $S_j$
while strengthening the existing obligations in $S_j$.
\end{proof}

\paragraph{$\psi = \tilde{\A}_{\geq c}$}
Let $\C^\psi = \langle \Sigma, S, s_0, X, \transitions, \F \rangle$ be defined as follows
(for simplicity, assume that $c > 0$):
\begin{itemize}
\item Each location $s \in S$ is of the form $\langle S_1, T_1 \dots, S_{m+1}, T_{m+1} \rangle$ where $S_j, T_j \subseteq S^\A$ for all $j \in \{1, \dots, m + 1\}$; intuitively, $\langle S_j, T_j \rangle$ can be seen as a location
of the sub-component $C^\psi_j$;
\item $s_0 = \langle \emptyset, \emptyset, \dots, \emptyset, \emptyset \rangle$;
\item $X = \{x_1, \dots, x_{m+1}\}$;
\item $\F = \emptyset$, i.e.~any run is accepting;
\item $\transitions$ is obtained by synchronising (in the same way as before) 
$\langle \ell, S, T \rangle \xrightarrow{\sigma, g, \lambda} \langle \ell', S', T' \rangle$
of individual sub-components:
	\begin{itemize}
	\item $p_\psi \notin \sigma$; $S' = \emptyset$, $S = \emptyset$, $T' = \emptyset$, $T = \emptyset$; $g = \top$; $\lambda = \emptyset$.

	\item $p_\psi \notin \sigma$; $S' = \emptyset$, $S = \emptyset$, $T \xrightarrow[\wedge]{\sigma} T'$, $T' \cap F^\A = \emptyset$; $g = \top$; $\lambda = \emptyset$.

	\item $p_\psi \notin \sigma$; $S \xrightarrow[\wedge]{\sigma} S'$, $T' = \emptyset$, $T = \emptyset$; $g = x < c$; $\lambda = \emptyset$.

	\item $p_\psi \notin \sigma$; $S \xrightarrow[\wedge]{\sigma} S'$, $T \xrightarrow[\wedge]{\sigma} T'$, $T' \cap F^\A = \emptyset$; $g = x < c$; $\lambda = \emptyset$.

	\item $p_\psi \notin \sigma$; $S' = \emptyset$, $T'$ is the union of some $T''$ such that $S \xrightarrow[\wedge]{\sigma} T''$
and $T'''$ with $T \xrightarrow[\wedge]{\sigma} T'''$, $T' \cap F^\A = \emptyset$; $g = x \geq c$; $\lambda = \emptyset$.

	\item $p_\psi \in \sigma$; $\{s_0^\A\} \xrightarrow[\wedge]{\sigma} S'$, $S = \emptyset$, $T' = \emptyset$, $T = \emptyset$; $g = \top$; $\lambda = \{ x \}$. 

	\item $p_\psi \in \sigma$; $\{s_0^\A\} \xrightarrow[\wedge]{\sigma} S'$, $S = \emptyset$, $T \xrightarrow[\wedge]{\sigma} T'$, $T' \cap F^\A = \emptyset$; $g = \top$; $\lambda = \{ x \}$. 

	\item $p_\psi \in \sigma$; $S'$ is the union of
some $S''$ such that $S \xrightarrow[\wedge]{\sigma} S''$ and $S'''$ with $\{s_0^\A\} \xrightarrow[\wedge]{\sigma} S'''$, $T' = \emptyset$, $T = \emptyset$; $g = x < c$; $\lambda = \emptyset$. 

	\item $p_\psi \in \sigma$; $S'$ is the union of
some $S''$ such that $S \xrightarrow[\wedge]{\sigma} S''$ and $S'''$ with $\{s_0^\A\} \xrightarrow[\wedge]{\sigma} S'''$, 
$T \xrightarrow[\wedge]{\sigma} T'$, $T' \cap F^\A = \emptyset$; $g = x < c$; $\lambda = \emptyset$. 

	\item $p_\psi \in \sigma$; $\{s_0^\A\} \xrightarrow[\wedge]{\sigma} S'$, $T'$ is the union of some $T''$ such that $S \xrightarrow[\wedge]{\sigma} T''$
and $T'''$ with $T \xrightarrow[\wedge]{\sigma} T'''$, $T' \cap F^\A = \emptyset$; $g = x \geq c$; $\lambda = \{x \}$. 
	\end{itemize}
\end{itemize}
\begin{proposition}
$\sem{\C^\psi} = \sem{\globally \big(p_\psi \Rightarrow \tilde{\A}_{\geq c}(\overline{\phi_1}, \dots, \overline{\phi_n})\big)}$.
\end{proposition}
\begin{proof}[Proof sketch.]
As in the proof of Proposition~\ref{prop:thirdcase}, whenever $p_\psi$ is pulled in $C^\psi$,
we attempt to put the corresponding obligation into an unused sub-component (with $S = \emptyset$)
or a sub-component that can be cleared (if $x > c$, we move the obligations in $S$ to $T$ and let them remain there).
If this is not possible, since for every $s^\A \in S^\A$ all the subtrees rooted at nodes
$(s^\A, v, i)$ can be replaced with the subtree rooted at 
$(s^\A, v^\textit{max}, i)$ where $v^\textit{max} = \max \{ v \mid (s^\A, v) \in L_i \}$,
some $C^\psi_j$ becomes redundant, and the obligations in the sub-component $C^\psi_k$ with the minimal non-negative
value of $x_k - x_j$ can be merged with the obligations in $C^\psi_j$, freeing up a sub-component for
the current incoming obligation.
This can be repeated to construct an accepting run of $\C^\psi$ on $\rho$.
The other direction holds as each $C^\psi_j$
adds an obligation that is at least as strong to $S_j$
without weakening the existing obligations in $S_j$.
\end{proof}

Finally, thanks to the fact that each location of $C^\psi$
can be represented using space polynomial in the size of $\A$,
and the product $\C^\textit{init}\times \prod_{\psi\in\Phi}\C^\psi$ need not to be constructed
explicitly, we can state the main result of this section.

\begin{theorem}\label{thm:pspace}
The satisfiability and model-checking problems for $\emitl_{0, \infty}$ over timed words
are $\mathrm{PSPACE}$-complete.
\end{theorem}

\begin{corollary}\label{cor:eeclpspace}
The satisfiability and model-checking problems for $\eecl$ over timed words
are $\mathrm{PSPACE}$-complete.
\end{corollary}

\section{Conclusion}

It is shown that $\emitl_{0, \infty}$ and \eecl{} are already as expressive as
$\emitl$ over timed words, a somewhat unexpected yet very pleasant result.
We also provided a compositional construction
from $\emitl_{0, \infty}$ to diagonal-free \ta{s}
based on one-clock alternating timed automata (\ocata{s});
this allows satisfiability and model checking 
based on existing algorithmic back ends for \ta{s}.
The natural next step would be to implement the construction
and evaluate its performance on real-world use cases.
Another possible future direction is to investigate whether
similar techniques can be used to handle full \emitl{} or
larger fragments of \ocata{s} (like~\cite{Ferrere18}).

\begin{acks}
The author would like to thank
Thomas Brihaye, Chih-Hong Cheng, Thomas Ferr\`{e}re, Gilles Geeraerts, Timothy M. Jones, Arthur Milchior, and Benjamin Monmege
for their help and fruitful discussions.
The author would also like to thank the anonymous reviewers for their comments. This work is supported by the \grantsponsor{gsidepsrc}{Engineering and Physical Sciences Research Council (EPSRC)}{https://epsrc.ukri.org/}
through grant \grantnum{gsidepsrc}{EP/P020011/1} and
(partially) by \grantsponsor{gsidfrsfnrs}{F.R.S.-FNRS}{http://www.fnrs.be/} PDR
grant \grantnum{gsidfrsfnrs}{SyVeRLo}.

\end{acks}

\bibliographystyle{ACM-Reference-Format}
\bibliography{refs} 

\appendix
\section{Proof of Proposition~\ref{prop:emitl2ocata}}

\begin{proof}
We first show that $\sem{\A_\phi} \subseteq \sem{\phi}$. Suppose that
$\A_\phi = \langle \Sigma, S, s_0, \atransitions, F \rangle$ has an accepting run $G = \langle V, \to \rangle$ on $\rho = (\sigma_1,\tau_1)(\sigma_2,\tau_2)\cdots$ and let 
$L_i = \{ (s, v) \mid (s, v, i) \in V \}$.
We claim that:
\begin{enumerate}

\item For each $s^\A \in S^\A$ where $\A$ occurs in a subformula $\A_I(\phi_1, \dots, \phi_n)$
and $i \geq 1$, $L_i \models_v s^\A$ implies that 
there is a \emph{finite} rooted DAG $G' = \langle V', \to' \rangle$ with $V' \subseteq S^\A \times \N_{\geq 0}$,
$(s^\A, i)$ as the root, and for each vertex $(s, k)$ either (a) $k > i$, $v + (\tau_k - \tau_i) \in I$, and $s \in F^\A$, or (b) there is some $a_{k+1} \in \Sigma^\A$
such that $L_{k+1} \models_{v + (\tau_{k+1} - \tau_i)} \atransitions(\phi_{a_{k+1}}, \sigma_{k+1})$
and there is a model $M$ of the
formula $\atransitions^\A(s, a_{k+1})$ such that $(s,k) \to' (s',k+1)$
\mbox{for every state $s'$ in $M$}.

\item For each $s^\A \in S^\A$ where $\A$ occurs in a subformula $\tilde{\A}_I(\phi_1, \dots, \phi_n)$
and $i \geq 1$, $L_i \models_v s^\A$ implies that 
there is a rooted DAG $G' = \langle V', \to' \rangle$ with $V' \subseteq S^\A \times \N_{\geq 0}$,
$(s^\A, i)$ as the root, and each vertex $(s, k)$ satisfies (a) if $k > i$ then $v + (\tau_k - \tau_i) \in I$ implies $s \notin F^\A$, and (b) for every $a_{k+1} \in \Sigma^\A$,
either $L_{k+1} \models_{v + (\tau_{k+1} - \tau_i)} \atransitions(\phi_{a_{k+1}}, \sigma_{k+1})$
or there is a model $M$ of the
formula $\overline{\atransitions^\A}(s, a_{k+1})$ such that $(s,k) \to' (s',k+1)$
for every state \mbox{$s'$ in $M$}.

\end{enumerate}

We prove (1); (2) can be proved similarly.
First note that every branch labelled with locations of $\A$ must terminate. 
If $(s^\A, v, i)$ is a leaf with respect to $S^\A$ in $G$, we must have $N \models_{v + (\tau_{i+1} - \tau_i)} \atransitions(\phi_a , \sigma_{i+1}) \land \atransitions^\A[s_F^\A \leftarrow s^\A_F \lor x \in I](s^\A, a)$ for some $a \in \Sigma^\A$ and $N \subseteq (S \setminus S^\A) \times \R_{\geq 0}$
such that $N \subseteq L_{i+1}$.
It follows that either (i) $N \models_{v + (\tau_{i+1} - \tau_i)} \atransitions(\phi_a , \sigma_{i+1}) \land \atransitions^\A(s^\A, a)$
or (ii) $v + (\tau_{i+1} - \tau_i) \in I$ and $N \cup M \models_{v + (\tau_{i+1} - \tau_i)} \atransitions(\phi_a , \sigma_{i+1}) \land \atransitions^\A(s^\A, a)$
for some nonempty $M \subseteq F^\A \times \{ v + (\tau_{i+1} - \tau_i) \}$.
In case (i), (b) trivially holds.
In case (ii) note that $N \models_{v + (\tau_{i+1} - \tau_i)} \atransitions(\phi_a , \sigma_{i+1})$, and let $\{ (s', i+1) \mid (s', v') \in M \}$ be the successors
of $(s^\A, i)$ in $G'$: each of them satisfies (a).
If $(s^\A, v, i)$ is not a leaf with respect to $S^\A$ in $G$, we have $N \cup M \models_{v + (\tau_{i+1} - \tau_i)} \atransitions(\phi_a , \sigma_{i+1}) \land \atransitions^\A[s_F^\A \leftarrow s^\A_F \lor x \in I](s^\A, a)$ for some $a \in \Sigma^\A$, $N \subseteq (S \setminus S^\A) \times \R_{\geq 0}$
such that $N \subseteq L_{i+1}$, and nonempty $M \subseteq S^\A \times \{ v + (\tau_{i+1} - \tau_i) \}$
such that $M \subseteq L_{i+1}$. It follows that either (i) $N \cup M \models_{v + (\tau_{i+1} - \tau_i)} \atransitions(\phi_a , \sigma_{i+1}) \land \atransitions^\A(s^\A, a)$
or (ii) $v + (\tau_{i+1} - \tau_i) \in I$ and $N \cup M' \models_{v + (\tau_{i+1} - \tau_i)} \atransitions(\phi_a , \sigma_{i+1}) \land \atransitions^\A(s^\A, a)$
for some $M' \supset M$ with $M' \setminus M \subseteq F^\A \times \{ v + (\tau_{i+1} - \tau_i) \}$.
In case (i), let $\{ (s', i+1) \mid (s', v') \in M \}$ be the successors
of $(s^\A, i)$ in $G'$; applying the IH on $L_{i+1} \models_{v + (\tau_{i+1} - \tau_i)} s'$
for all such $s'$ yields the corresponding sub-DAGs, which we combine to obtain $G'$.
Case (ii) is similar.

We now claim that for each subformula $\psi$ of $\phi$ and  $i \geq 1$,
$L_i  \models_0 \atransitions(\psi, \sigma_i)$ implies $\rho, i \models \psi$.
The cases of atomic propositions, negation, and Boolean operators
are trivial. For $\psi = \A_I(\phi_1, \dots, \phi_n)$, as $L_i \models_0 \atransitions(s_0^\A, \sigma_i)$
we have, for some $a \in \Sigma^\A$, (i) $L_i \models_0 \atransitions(\phi_a, \sigma_i)$
and (ii) $L_i \models_0 \atransitions^\A[s_F^\A \leftarrow s^\A_F \lor x \in I](s_0^\A, a)$.
From (i) and the IH we obtain $\rho, i \models \phi_a$.
From (ii) we deduce that there is $M \subseteq (S^\A \setminus F^\A) \times \{ 0 \}$
such that $M \subseteq L_i$, and $M' \supseteq M$ such that
$M' \setminus M \subseteq F^\A \times \{ 0 \}$ and $M' \models_0 \atransitions^\A(s_0^\A, a)$.
Applying (1) on every $s^\A \in M$ and the IH (note that by the definition of $\atransitions$,
we have $L_{k+1} \models_0 \atransitions(\phi_{a_{k+1}}, \sigma_{k+1})$ if $L_{k+1} \models_{v + (\tau_{k+1} - \tau_i)} \atransitions(\phi_{a_{k+1}}, \sigma_{k+1})$)
and combining sub-DAGs gives $\rho, i \models \psi$.
The case of $\psi = \tilde{\A}_I(\phi_1, \dots, \phi_n)$  is analogous.

Finally, as $s^\textit{init} \in L_0$, we have $L_1 \models_{\tau_1} x.\atransitions(\phi, \sigma_1)$,
which is equivalent to $L_1 \models_0 \atransitions(\phi, \sigma_1)$. It follows that $\rho, 1 \models \phi$; this finishes the proof of $\sem{\A_\phi} \subseteq \sem{\phi}$.

For the other direction, observe that $\A_{\neg \phi} = (\A_\phi)^c$ (up to renaming of locations)
where $(\A_\phi)^c$ is obtained from $\A_\phi$ by dualising the transition function and 
swapping the sets of final and non-final locations.
The desired result ($\sem{\phi} \subseteq \sem{\A_\phi}$)
immediately follows from the fact that $\A_\phi$ is a \emph{tree-like} \ocata{}~\cite{BriEst14}.
\end{proof}




\end{document}